\documentclass[journal,twocolumn]{IEEEtran}
\usepackage{ifpdf}
\usepackage{cite}
\usepackage{graphicx} 
\usepackage{mathtools}
\usepackage{amsmath,amsfonts,bm,amssymb}
\usepackage{mathalfa}
\usepackage[thinc]{esdiff}
\usepackage{balance}
\usepackage{tikz}
\usepackage{array}
\usepackage{xcolor}
\usepackage[caption=false,font=footnotesize]{subfig}

\newcommand{\lr}[1]{\langle #1 \rangle}
\newcommand{\blr}[1]{\big\langle #1 \big\rangle}
\newcommand{\bblr}[1]{\bigg\langle #1 \bigg\rangle}

\newcommand{\ji}{{\sf j}}

\newtheorem{theorem}{Theorem}

\newtheorem{corollary}{Corollary}
\newtheorem{remark}{Remark}

\newcommand{\bs}[1]{\boldsymbol{#1}}
\newcommand{\ms}[1]{\mathsf{#1}}
\newcommand{\mc}[1]{\mathcal{#1}}

\newcommand{\mb}[1]{\mathbf{#1}}
\newcommand{\mr}[1]{\mathrm{#1}}
\newcommand{\tr}{\mathrm{Tr}}
\newcommand{\var}{\mathrm{Var}}

\newcommand{\md}{\mathrm{d}}
\newcommand{\e}{\mathrm{e}}

\usepackage{algorithm, algorithmicx, algpseudocode}

\title{Variational Bayesian Data Detection for Multiuser MIMO Systems Corrupted by Phase Noises}
\author{Toan-Van Nguyen and Duy H. N. Nguyen
\thanks{The authors are with the Department of Electrical and Computer Engineering, San Diego State University, San Diego, CA 92182 USA (emails:\{tnguyen58, duy.nguyen\}@sdsu.edu).}}

\begin{document}

	\maketitle
	
\begin{abstract}
Phase noise (PN), arising from imperfect local oscillators, introduces multiplicative distortions that degrade the performance of communication systems. In uplink multiuser multiple-input multiple-output (MIMO) systems, this impairment is further compounded by the presence of independent oscillators at each transmit and receive antenna, each contributing an uncorrelated noise component. Existing PN compensation algorithms at the receiver either rely on linearization approximations that lose accuracy under severe PN conditions, or incur computational complexity that scales prohibitively with the number of antennas. To address these limitations, we propose a variational Bayes (VB) framework for joint PN estimation and data detection in uplink MIMO systems. We develop VB-based detectors that treat noise statistics as latent variables, and reformulate the inference problem by absorbing the transmitter PN into the transmitted signal, treating the resulting composite variable as the inference target. Under von Mises priors, this reformulation yields exact closed-form conjugate posterior updates, from which we derive an improved detector achieving superior performance at low complexity. Simulation results demonstrate that the proposed VB algorithm achieves lower symbol error rates than the Self-Interference Whitening (SIW) algorithm and conventional phase-noise-unaware detectors across a wide range of channel conditions, modulation orders, and PN severities, while remaining computationally scalable to large MIMO deployments.
\end{abstract}
\section{Introduction}
The imperfect local oscillators at the transmitter and receiver introduce phase noise (PN), a major hardware impairment in multiuser multiple-input multiple-output (MIMO) systems \cite{chang2024modeling}. The demand for higher spectral efficiency drives communication systems toward high-order constellations. However, even moderate PN can induce severe rotational distortion in high-order received signal constellations, resulting in pronounced error floors at high signal-to-noise ratio (SNR) \cite{pitarokoilis2014uplink}. This problem becomes severe at millimeter-wave (mmWave) and terahertz frequencies, where oscillator instability is higher, as well as in geostationary satellite links, where large phase drifts caused by the Doppler effect exacerbate the distortion \cite{marrero2023accurate}. In massive MIMO uplink systems, the challenge is further compounded because each user and each receive antenna may be driven by an independent oscillator, resulting in a multiplicative, high-dimensional PN process that is tightly coupled with inter-user interference and the non-Gaussian discrete symbol prior, rendering exact maximum-likelihood (ML) detection computationally intractable \cite{rasekh2021phase}. Addressing this challenge requires advanced signal processing methods that can achieve accurate, scalable, and amenable to closed-form computation.

Early PN mitigation methods relied on linearization of the exponential PN term via first-order Taylor expansion. The Extended Kalman Filter-Smoother (EKFS) \cite{nasir2013phase,reggiani2018extended,shehata2009joint} and the Expectation Maximization (EM) algorithm \cite{isikman2014joint,tarable2013based} were proposed, where EKFS applied Kalman filtering on the resulting linear-Gaussian model, while EM alternated between PN estimation and symbol likelihood maximization, both achieving accuracy under mild PN and coded MIMO systems.
A variational Bayesian (VB) framework in \cite{nissila2009adaptive} generalized EM by maintaining a full phase posterior rather than a point estimate, and handled dynamic PN via an Extended Kalman Smoother. However, all linearization-based methods share a common limitation: the Taylor approximation becomes increasingly inaccurate when the PN variance increases, leading to performance degradation in the high-PN regimes encountered at high carrier frequencies.

To handle strong PN without linearization, factor-graph methods operating on the full probabilistic model have been proposed. Tikhonov (also known as von Mises) parametrization within the sum-product algorithm (SPA) was proposed to enable accurate PN tracking for additive white Gaussian noise (AWGN) and optical channels \cite{colavolpe2005algorithms,alfredsson2019iterative}, while expectation propagation (EP)-based variants reduced pilot overhead on intersymbol interference (ISI) channels \cite{conti2025application}.
Hybrid receivers combining belief propagation (BP), mean-field, and EP have extended this to coded ISI and multiuser MIMO settings \cite{wang2016bp,he2022auxiliary}, though their complexity scales cubically with the channel memory length.
Variational message passing \cite{wang2016low,badiu2013message} and SPA-based factor-graph detectors \cite{krishnan2015algorithms} offered complementary approaches by restricting all message updates to closed-form exponential family distributions, thereby handling the nonlinear PN observation model without linearization.

Approximate ML and Bayesian methods offer more scalable alternatives. In \cite{combes2017approximate}, the self-interference whitening (SIW) approach approximated the ML likelihood by pre-whitening the PN-induced self-interference, which reduced detection to a minimum Euclidean distance problem solvable by sphere decoding. The SIW avoids the PN tracking process, but its first-order PN approximation becomes inaccurate at large PN variance. The sphere-decoding in SIW also incurs complexity that grows cubically with the number of antennas \cite{vikalo2005sphere}, making the SIW impractical for large MIMO systems. The generalized expectation-consistent signal recovery (GEC-SR) framework in \cite{he2017generalized} replaced sphere decoding with an iterative message-passing procedure whose per-iteration cost was dominated by matrix-vector products, significantly reducing complexity \cite{yang2019symbol}. However, GEC-SR relies on approximating all involved distributions as Gaussian, which might lose validity at high PN variance and cause performance degradation when the antenna array increases. For PN-free MIMO systems, VB inference has been shown to offer an attractive balance between detection accuracy and computational efficiency. Particularly, the VB inference was developed for data detection in \cite{nguyen2025mimo,nguyen2024optimal,do2025variational} and joint data detection and channel estimation in \cite{nassirpour2025variational}. The goal of VB is to find an approximation for the true posterior distribution of latent variables, for example, data symbols, PNs, and noise variance in this work, given the observed data \cite{nguyen2022variational}. VB has also been used in the context of joint channel estimation and data detection in massive MIMO \cite{nassirpour2025variational} and DoA estimation \cite{nguyen2025gridless}.

\subsection{Motivations and Contributions}
These aforementioned works reveal a trade-off between accuracy and tractability. Linearization-based methods are efficient but fail under strong PN \cite{nasir2013phase,reggiani2018extended,shehata2009joint}. Factor-graph methods are more accurate but impose high complexity, especially in MIMO with several independent oscillators \cite{colavolpe2005algorithms,alfredsson2019iterative,conti2025application}. Approximate ML and Bayesian approaches reduce complexity but introduce posterior mismatches that widen the performance gap at high SNR \cite{combes2017approximate}. In this paper, we develop VB algorithms for joint PN estimation and data detection in multiuser MIMO systems. Different from \cite{nissila2009adaptive,conti2025application,wang2016bp}, which mainly focus on AWGN channels and incur high computational complexity, our VB framework yields closed-form expressions for the posterior distributions of the PNs and data symbols, with per-iteration cost that scales linearly in the number of users, receive antennas, and constellation size. 
Conventional separation-based approaches treat the transmitter PN as an independent latent variable estimated alongside the data symbols. However, the data symbol and the transmitter PN are coupled through the composite transmit signal and cannot be accurately approximated as independent, so this separation introduces a posterior mismatch that degrades detection performance. In this paper, we resolve this by absorbing the transmitter PN directly into the transmitted signal and treating the resulting composite variable as the inference target. 

The main contributions of this paper are as follows:
\begin{itemize}
  \item We develop a matched-filter variational Bayes (MF-VB) algorithm for joint PN estimation and data detection in uplink multiuser MIMO systems, adopting von Mises priors for transmitter and receiver PNs and treating noise precision as a latent variable. Closed-form posterior updates are derived via coordinate-ascent variational inference (CAVI).

\item We further develop an LMMSE-VB detector that generalizes MF-VB by modeling the full noise covariance matrix as the inference target, capturing residual inter-user interference structure and improving robustness under correlated fading and mmWave channels.

\item We resolve a posterior mismatch in separation-based approaches by absorbing the transmitter PN into the transmitted signal, treating the composite variable as the inference target. Under a von Mises prior, this yields an exact closed-form conjugate posterior as a finite mixture of von Mises distributions. The resulting improved MF-VB detector eliminates separate PN tracking variables, achieving superior detection and PN estimation performance at the same complexity as MF-VB.

\item We provide a complexity analysis showing that MF-VB and improved MF-VB scale linearly in the number of users, receive antennas, constellation size, and iterations, contrasting with the exponential scaling of ML and SIW detectors. Simulations confirm that the improved MF-VB achieves superior detection performance, particularly in the high-SNR regime.
	\end{itemize}
The rest of the paper is organized as follows. Section~\ref{sec:Sys} introduces the signal model for the PN estimation problem and preliminary results. Section~\ref{sec:VBforMIMO} analyzes the VB framework for data detection with PNs and the scalar and full noise covariance models, leading to the MF-VB and LMMSE-VB algorithms, respectively. Section~\ref{sec:ImprovedMF} proposes the reformulated signal model that absorbs the transmitter PN into the composite transmit variable, and derives the improved MF-VB algorithm along with a computational complexity analysis of all proposed detectors. Simulation results are provided in Section~\ref{sec:Sim}, where the proposed algorithms are evaluated under i.i.d. and correlated Rayleigh fading channels, mmWave propagation environments. Section~\ref{sec:Conclusions} concludes the paper.

\underline{\textit {Notation}}: Boldface lowercase and boldface uppercase variables denote vectors and matrices, respectively. The $L_2$-norm and the absolute value are indicated by $\|\cdot\|$ and $|\cdot|$, respectively. Real and imaginary parts are denoted by $\Re\{\cdot\}$ and $\Im\{\cdot\}$, respectively, with $\ji =\sqrt{-1}$. The distribution of a complex multivariate Gaussian $\mb{x}$ with mean $\bs{\mu}$ and covariance matrix $\bs{\Sigma}$ is denoted by $\mc{CN}(\mb{x};\bs{\mu},\bs{\Sigma})$. 
and is also written as $\mb{x} \sim \mathcal{CN}(\boldsymbol{\mu}, \bs{\Sigma})$. Similarly, the distribution of a von Mises random variable $x$ with natural parameter $\eta$ is denoted by $\mc{VM}(x;\eta)$, and is also written as $x\sim \mc{VM}(\eta)$. The identity matrix is denoted by $\mb{I}$, the trace operator by $\tr$($\cdot$), and the expectation operator by $\mathbb{E}\{\cdot\}$. The transpose, complex conjugate, and complex conjugate transpose operators are denoted by $(\cdot)^T$, $(\cdot)^*$, and $(\cdot)^H$, respectively. The symbols $\propto$ and $\sim$ stand for ``is proportional to'' and ``distributed according to'', respectively. The bracket $\lr{\cdot}$ represents expectation with respect to all latent variables except the one under consideration. 

\section{System Model and Problem Formulation} \label{sec:Sys}
	
	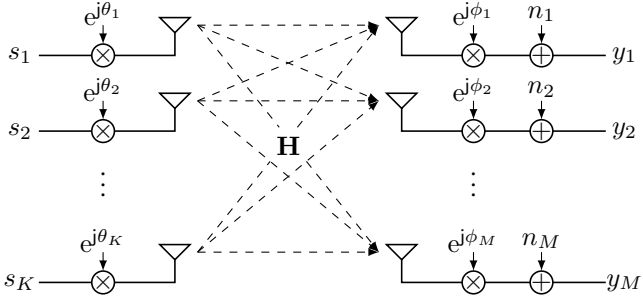
\begin{figure}
		\centering
		\begin{tikzpicture}
			\node at (-0.55,4) {$s_{1}$};
			\draw [semithick] (-0.3,4) to (0.4,4);
			\draw [semithick] (0.55,4) circle [radius=0.15] node {$\times$};
			\draw [-latex] (0.55,4.4) to (0.55,4.15);
			\node at (0.555,4.55) {$\e^{\ji\theta_1}$};
			\draw [semithick] (0.7,4) to (1.5,4) to (1.5,4.3) to (1.3,4.5) to (1.7,4.5) to (1.5,4.3);
			
			\node at (-0.55,3) {$s_{2}$};
			\draw [semithick] (-0.3,3) to (0.4,3);
			\draw [semithick] (0.55,3) circle [radius=0.15] node {$\times$};
			\draw [-latex] (0.55,3.4) to (0.55,3.15);
			\node at (0.55,3.55) {$\e^{\ji\theta_2}$};
			\draw [semithick] (0.7,3) to (1.5,3) to (1.5,3.3) to (1.3,3.5) to (1.7,3.5) to (1.5,3.3);
			
			\node at (-0.55,1) {$s_{K}$};
			\draw [semithick] (-0.3,1) to (0.4,1);
			\draw [semithick] (0.55,1) circle [radius=0.15] node {$\times$};
			\draw [-latex] (0.55,1.4) to (0.55,1.15);
			\node at (0.55,1.55) {$\e^{\ji\theta_K}$};
			\draw [semithick] (0.7,1) to (1.5,1) to (1.5,1.3) to (1.3,1.5) to (1.7,1.5) to (1.5,1.3);
			
			\node at (0.55,2.4) {$\vdots$};
			
			\node at (3,2.8) {$\mb{H}$};
			
			\draw [semithick] (5.3,4) to (4.5,4) to (4.5,4.3) to (4.3,4.5) to (4.7,4.5) to (4.5,4.3);
			\draw [semithick] (5.45,4) circle [radius=0.15] node {$\times$};
			\draw [-latex] (5.45,4.4) to (5.45,4.15);
			\node at (5.45,4.55) {$\e^{\ji\phi_1}$};
			\draw [semithick] (5.6,4) to (6.2,4);
			\draw [semithick] (6.35,4) circle [radius=0.15] node {$+$};
			\draw [-latex] (6.35,4.4) to (6.35,4.15);
			\node at (6.35,4.55) {$n_1$};
			\draw [semithick] (6.5,4) to (7.2,4);
			\node at (7.45,4) {$y_{1}$};
			
			\draw [semithick] (5.3,3) to (4.5,3) to (4.5,3.3) to (4.3,3.5) to (4.7,3.5) to (4.5,3.3);
			\draw [semithick] (5.45,3) circle [radius=0.15] node {$\times$};
			\draw [-latex] (5.45,3.4) to (5.45,3.15);
			\node at (5.45,3.55) {$\e^{\ji\phi_2}$};
			\draw [semithick] (5.6,3) to (6.2,3);
			\draw [semithick] (6.35,3) circle [radius=0.15] node {$+$};
			\draw [-latex] (6.35,3.4) to (6.35,3.15);
			\node at (6.35,3.55) {$n_2$};
			\draw [semithick] (6.5,3) to (7.2,3);
			\node at (7.45,3) {$y_{2}$};
			
			\draw [semithick] (5.3,1) to (4.5,1) to (4.5,1.3) to (4.3,1.5) to (4.7,1.5) to (4.5,1.3);
			\draw [semithick] (5.45,1) circle [radius=0.15] node {$\times$};
			\draw [-latex] (5.45,1.4) to (5.45,1.15);
			\node at (5.45,1.55) {$\e^{\ji\phi_M}$};
			\draw [semithick] (5.6,1) to (6.2,1);
			\draw [semithick] (6.35,1) circle [radius=0.15] node {$+$};
			\draw [-latex] (6.35,1.4) to (6.35,1.15);
			\node at (6.35,1.55) {$n_M$};
			\draw [semithick] (6.5,1) to (7.2,1);
			\node at (7.45,1) {$y_{M}$};	
			
			\node at (5.45,2.4) {$\vdots$};
			
			\draw [-latex,dashed] (1.8,4.4) to (4.2,4.4);
			\draw [-latex,dashed] (1.8,4.4) to (4.2,3.4);
			\draw [dashed] (1.8,4.4) to (2.88,3.05);
			\draw [-latex,dashed] (3.24,2.6) to (4.2,1.4);
			\draw [-latex,dashed] (1.8,3.4) to (4.2,4.4);
			\draw [-latex,dashed] (1.8,3.4) to (4.2,3.4);
			\draw [-latex,dashed] (1.8,3.4) to (4.2,1.4);
			\draw [-latex,dashed] (1.8,1.4) to (4.2,1.4);
			\draw [-latex,dashed] (1.8,1.4) to (4.2,3.4);	
			\draw [dashed] (1.8,1.4) to (2.76,2.6);
			\draw [-latex,dashed] (3.12,3.05) to (4.2,4.4);
		\end{tikzpicture}
		\caption{Diagram of a MIMO system with $K$ transmit antennas and $M$ receive antennas. The transmitted signal at user-$i$ is corrupted by the phase noise factor $\e^{\ji\theta_i}$, whereas the received signal at antenna-$m$ is corrupted by the phase noise factor $\e^{\ji\theta_m}$ and the background noise $n_m$.}
		\label{figure-model}
	\end{figure}
	
\subsection{System Description}
We consider an uplink MIMO system consisting of $K$ single-antenna users and a BS equipped with $M$ antennas, as shown in Fig. \ref{figure-model}. We denote $s_i$ as the transmitted symbol, drawn from a discrete constellation $\mc{A}$, \emph{e.g.}, quadrature amplitude modulation (QAM) or phase-shift keying (PSK). Each symbol is zero-mean with normalized power, i.e., $\mathbb{E}[s_i] = 0$ and $\mathbb{E}[|s_i|^2] = 1$. The prior distribution of $s_i$ is given by
	\begin{align}
		p(s_i) = \sum_{a \in \mc{S}}p_a\delta(s_i - a),
	\end{align}
where $p_a$ is the known prior probability of the constellation point $a$ and $\delta(s_i-a)$ indicates the point mass function at $a$. Each transmitted symbol $s_i$ is perturbed by a phase noise component $\theta_i$, such that the \textit{phase-noise-corrupted transmit signal} from user-$i$ is given by
\begin{align} \label{eq:x_i}
		x_i = s_i\, \e^{\ji\theta_i}.
\end{align}
Define the transmit PN matrix as $\bs{\Theta} = \mr{diag}\big(\e^{\ji \theta_1},\ldots,\e^{\ji \theta_K}\big)$ and let $\mb{x} = \bs{\Theta}\mb{s}$ be the PN-corrupted transmit vector, the received signal vector at the BS can be expressed as
	\begin{align}\label{eq:system_model}
		\mb{y} = \bs{\Phi}\mb{H}\bs{\Theta}\mb{s} + \mb{n} = \bs{\Phi}\mb{H}\mb{x} + \mb{n},
	\end{align}
where $\bs{\Phi} = \mr{diag}\big(\e^{\ji \phi_1},\ldots,\e^{\ji \phi_M}\big)$ captures the phase noises at the BS, the uplink channel matrix $\mb{H} = [\mb{h}_1,\ldots,\mb{h}_K] \in \mathbb{C}^{M \times K}$ is known at the BS, and $\mb{n}$ is the Gaussian noise vector comprising of i.i.d. $\mc{CN}(0,N_0)$ random variables. 

Unless otherwise stated, we assume the channel vector $\mb{h}_i \in \mathbb{C}^{M\times 1}$ associated with user $i$ follows a Gaussian distribution, $p(\mb{h}_i) = \mc{CN}(\mb{h}_i;\mb{0};\mb{R}_i)$, where $\mb{R}_i =\mathbb{E}[\mb{h}_i\mb{h}_i^H]$ denotes the channel covariance matrix. Note that $\mb{R}_i$ is generally not a scaled identity matrix; rather, it is typically modeled to capture the spatial correlation and large-scale fading between user $i$ and the BS. Denoting $\bs{\phi} = [\phi_1,\ldots,\phi_M]$ and $\bs{\theta} = [\theta_1,\ldots,\theta_K]$, the marginal conditional distribution $p(\mb{y}\!\mid\!\mb{s},\mb{H})$ is given by
	\begin{align} 
		p(\mb{y}\!\mid\!\mb{s};\mb{H}) &= \int_{\bs{\phi}}\int_{\bs{\theta}} p(\mb{y},\bs{\phi},\bs{\theta}\!\mid\!\mb{s};\mb{H})\md\bs{\theta}\md\bs{\phi} \nonumber \\
		& = \mathbb{E}_{\bs{\phi},\bs{\theta}} \big[p(\mb{y}\!\mid\!\mb{s},\bs{\phi},\bs{\theta};\mb{H})\big].
	\end{align}
	Assuming a uniform input distribution on $\mb{s}$, the optimal detection performance can be attained by maximizing the likelihood function $p(\mb{y}\!\mid\!\mb{s};\mb{H})$, 
    i.e.,
	\begin{align} \label{eq:ML:detector}
		\mb{s}^{\star} &= \arg\max_{\mb{s}\in\mc{A}^K} p(\mb{y}\!\mid\!\mb{s};\mb{H}) \nonumber \\
		&= \arg\max_{\mb{s}\in\mc{A}^K} \ln \mathbb{E}_{\bs{\phi},\bs{\theta}} \left[\e^{-N_0^{-1}\|\mb{y}-\bs{\Phi}\mb{H}\bs{\Theta}\mb{s}\|^2}\right].
	\end{align}
	However, the expectation in \eqref{eq:ML:detector} can not be obtained in closed-form since it involves a high-dimensional integration over the PN variables. It requires computing an integral over $K+M$ dimensions, which becomes extremely difficult for high dimensions in massive MIMO. Furthermore, the search space $|\mc{A}|^K$ grows prohibitively large with both the modulation size $|\mc{A}|$ and the number of users $K$, resulting in excessive detection complexity.

	\begin{remark}
		If all antennas at the BS share the same oscillator, the BS will be affected by a single source of phase noise, i.e, $\phi \triangleq \phi_1=\ldots=\phi_M$. This presents a simplified model of the system model in \eqref{eq:system_model}, with fewer parameters to estimate. Our proposed VB frameworks can be adopted for this common-oscillator scenario. 
	\end{remark}

\subsection{Preliminaries}
	To support the upcoming discussion, we provide some preliminary results from real analysis and probability, presented without proofs:
	\begin{itemize}
		\item For a complex number $z$,
		$$\Re\{\ji z\} = -\Im\{z\} = \Im\{z^*\}.$$
		\item For a von Mises RV $x \sim \mc{VM}(\eta)$, $\mu \triangleq \angle \eta$ is the mean direction and $\kappa \triangleq |\eta|$ is trob
        he concentration of $x$, analogous to the inverse of the variance. The PDF of $x$ is 
		$$\mc{VM}(x;\eta) = \frac{\exp\bigl(\kappa \cos(x-\mu)\bigr)}{2\pi I_0(\kappa)} \propto \exp\bigr(\Re\{\eta^* \e^{\ji x}\}\bigr).$$
		The 
        expectation and variance of $\e^{\ji n {x}}$ are
        \begin{align*} 
		    \mathbb{E}\big[\e^{\ji n{x}}\big] = \frac{I_{|n|} (\kappa)}{I_0(\kappa)} \e^{\ji n\mu}, \;\; \var\big[\e^{\ji n{x}}\big] = 1 - \left(\frac{I_{|n|}(\kappa)}{I_0(\kappa)}\right)^2,
		\end{align*}
		where $I_i(\kappa)$ is the modified Bessel function of the first kind of order $i$. For large value of $\kappa$ and specific order $n=1$, one can have $\frac{I_{1} (\kappa)}{I_{0} (\kappa)} \approx 1-\frac{1}{2\kappa}$.
		\item Let $x$ and $y$ be two independent RVs  with means $\hat{x}$ and $\hat{y}$ and variances $\sigma^2_{x}$ and $\sigma^2_{{y}}$, respectively. Then
		$$ \var\big[{x}{y}\big] = \big(|\hat{x}|^2 + \sigma^2_x\big) \big(|\hat{y}|^2 + \sigma^2_y\big)  - |\hat{x}|^2|\hat{y}|^2,$$
		and 
		$$ \var\big[\e^{\ji n{x}}{y}\big] = \var\big[\e^{\ji n {x}}\big]|\hat{y}|^2 + \sigma^2_y.$$
		\item Let $\mb{x} \in \!\mathbb{C}^{n\times 1}$ and $\mb{y}\in \!\mathbb{C}^{n\times 1}$ be independent RVs w.r.t. distribution $q_{\mb{x}, \mb{y}}(\mb{x},\mb{y}) = q(\mb{x})q(\mb{y})$. Denote $\hat{\mb{x}}$ and $\bs{\Sigma}_{\mb{x}}$ (respectively, $\hat{\mb{y}}$ and $\bs{\Sigma}_{\mb{y}}$) the variational mean and covariance of $\mb{x}$ (respectively, $\mb{y}$). For an arbitrary Hermitian matrix $\mb{B}$, let $\blr{(\mb{y}-\mb{A}\mb{x})^H\mb{B}(\mb{y}-\mb{A}\mb{x})}$ be the expectation of $(\mb{y}-\mb{A}\mb{x})^H\mb{B}(\mb{y}-\mb{A}\mb{x})$ w.r.t. $q_{\mb{x},\mb{y}}(\mb{x},\mb{y})$, then
		\begin{align} \label{f-ABC}
			&\blr{(\mb{y}-\mb{A}\mb{x})^H\mb{B}(\mb{y}-\mb{A}\mb{x})} = \big(\hat{\mb{y}}-{\mb{A}}\hat{\mb{x}}\big)^H\mb{B}\big(\hat{\mb{y}}-{\mb{A}}\hat{\mb{x}}\big) \nonumber \\
			&\qquad\qquad\qquad\quad + \tr\{\mb{B}\bs{\Sigma_{\mb{y}}}\} + \tr\big\{\bs{\Sigma}_{\mb{x}}{\mb{A}^H}\mb{B}{\mb{A}} \big\}.
		\end{align}
	\end{itemize}

\section{Variational Bayes for MIMO Detection with Phase Noises} \label{sec:VBforMIMO}
In this section, we develop VB methods to jointly infer the transmitted symbols $\mb{s}$, phase noises $\bs{\theta}$ and $\bs{\phi}$, from the observation $\mb{y}$, assuming a known channel $\mb{H}$. In practice, the noise variance/covariance is not known a priori and needs to be estimated. To this end, we estimate the input signal $\mb{x}$ considering two scenarios of postulated noise models: The postulated noise variance $N_0^{\rm post}$ and covariance $\mb{C}^{\rm post}$.
\subsection{MF-VB Algorithm} \label{sec:MF:von}
We assume a Gamma distribution $p(\gamma) = {\rm Gamma}(\alpha,\beta)$, as a conjugate prior for the precision parameter $\gamma$, where $\alpha$ and $\beta$ are the shape and rate parameters, respectively. 
Considering the precision $\gamma$ as a random variable, the conditional probability of observation $\mb{y}$ and latent variables $\mb{s},\bs{\phi}$, and $\bs{\theta}$ can be factorized as
\begin{align} \label{eq:con:gau:mf}
			p(\mb{y},\mb{s},\bs{\theta},\bs{\phi},\gamma) = p(\mb{y}\!\mid\!\mb{s},\bs{\theta},\bs{\phi},\gamma) p(\mb{s})p(\bs{\theta})p(\bs{\phi})p(\gamma),
\end{align}
where $p(\mb{y}\,|\,\mb{s},\bs{\phi},\bs{\theta},\gamma) = \mc{CN}\big(\bs{\Phi}^H\mb{y};\mb{H}\bs{\Theta}\mb{s} ,\gamma^{-1}\mb{I}_M\big)$.
We approximate the posterior using a mean-field factorization as
\begin{align}	p(\mb{s},\bs{\theta},\bs{\phi},\gamma\!\mid\!\mb{y};\mb{H}) &\approx  q(\mb{s},\bs{\theta},\bs{\phi},\gamma) \nonumber \\
&= \Bigg[\prod_{i=1}^K q(s_i)q({\theta}_i)\Bigg]  \prod_{m=1}^M q({\phi}_m)q(\gamma).
\end{align}
The variational factors are obtained by minimizing the Kullback-Leibler divergence to the true posterior. For the variational distribution of a random variable of interest $x$, the coordinate ascent update for the variational is given by
\begin{align}
		\label{eq:q_start_prop}
		q^{\star}(x) &\propto \mr{exp}\left\{\big\langle{\ln p (\mb{y},\mb{s}, \bs{\theta},\bs{\phi},\gamma)\big\rangle_{-x}}\right\},
	\end{align}
where the expectation is taken with respect to all variables except $x$. To enable tractable updates, we assume von Mises priors for $\theta_i$ and $\phi_m$, given by\footnote{Phase noise (increment) is typically modeled as Gaussian $\mc{N}(0,\gamma^{-1})$. For small innovation variance $\gamma^{-1}$, the phase noise (increment) can be indistinguishably modeled as von Mises $\mc{VM}(\gamma)$.}
\[p(\theta_i) = \mc{VM}(\theta_i;\varepsilon_{\mr{t},i}),\quad 
p(\phi_m) = \mc{VM}(\phi_m;\varepsilon_{\mr{r},m}).\]

For a \emph{current} estimate of the variational distributions $q(\mb{s})$, $q(\bs{\theta})$, $q(\bs{\phi})$, and $q(\gamma)$, we will obtain the corresponding variational means $\hat{\mb{s}}$, $\hat{\bs{\theta}}$,  $\hat{\bs{\phi}}$,  and $\hat{\gamma}$, respectively; and likewise $\hat{\bs{\Theta}}$ and $\hat{\bs{\Phi}}$. For ease of presentation and computational complexity, we define a residual term $\mb{r} =  \hat{\bs{\Phi}}^H\mb{y}  - \mb{H}\hat{\bs{\Theta}}\hat{\mb{s}}$ using these \emph{current} corresponding estimates and use it for updating the variational distributions in the following iteration.

	\emph{1) Update $\phi_m$:} Taking the expectation of the conditional in \eqref{eq:con:gau:mf} w.r.t. all latent variables except $\phi_m$, the variational distribution $q(\phi_m)$ can be derived as
	\begin{align}
		q(\phi_m) &\propto p(\phi_m)\exp\Big\{-\hat{\gamma}\blr{|\e^{-\ji \phi_m}y_m - \mb{H}_{m,:}\bs{\Theta}\mb{s}|^2}_{-\phi_m}\Big\}\nonumber\\
		&\propto p(\phi_m)\exp\Big\{\hat{\gamma}\blr{ 2\Re\{\e^{\ji \phi_m}y_m^*\mb{H}_{m,:}\bs{\Theta}\mb{s}\}}_{-\phi_m}\Big\}\nonumber\\
		&\propto p(\phi_m)\exp\bigl(\Re\{\eta_{\mr{r},m}^*\e^{\ji \phi_m}\}\bigr),
	\end{align}
	where 
	$\eta_{\mr{r},m} = 2\hat{\gamma}y_m\big(\mb{H}_{m,:}\hat{\bs{\Theta}}\hat{\mb{s}}\big)^*$.
	
	Since the prior $p(\phi_m)$ is assumed to be von Mises $\mc{VM}(\varepsilon_{\mr{r},m})$, $q(\phi_m)$ is also von Mises with mean $\hat{\phi}_m = \angle(\varepsilon_{\mr{r},m} + \eta_{\mr{r},m})$ and concentration $\kappa_{\mr{r},m} = |\varepsilon_{\mr{r},m} + \eta_{\mr{r},m}|$. Thus, we have $\blr{\e^{\pm \ji \phi_m}} = d_{\mr{r},m} \e^{\pm \ji \hat{\phi}_m}$, where $d_{\mr{r},m} = {I_1(\kappa_{\mr{r},m})}/{I_0(\kappa_{\mr{r},m})}$. As a result, $\eta_{\mr{r},m}$ can be rewritten as  
	\begin{align}
		\eta_{\mr{r},m} &= 2\hat{\gamma} y_m \big(\blr{\e^{-\ji{\phi}_m}}y_m - r_m\big)^* \nonumber \\
		&= 2\hat{\gamma} \big(|y_m|^2 d_{\mr{r},m} \e^{\ji\hat{\phi}_m} - y_mr_m^*\big). 
	\end{align} 
	
	
	
	\emph{2) Update $\theta_i$:} Taking the expectation of the conditional in \eqref{eq:con:gau:mf} w.r.t. all latent variables except $\theta_i$, the variational distribution $q(\theta_i)$ can be derived as
    \vspace{-0.1cm}
	\begin{align}
		q(\theta_i) &\propto p(\theta_i) \exp\!\Big\{-\hat{\gamma}\blr{\|\bs{\Phi}^H\mb{y} - \mb{H}\bs{\Theta}\mb{s}\|^2}_{-\theta_i}\Big\} \nonumber \\
		&\propto p(\theta_i) \exp\!\bigg\{2\hat{\gamma}\Re\bigg\{ \!\bigg(\!\hat{\bs{\Phi}}^H\mb{y} \!- \!\sum_{j\neq i}^K\! \mb{h}_j\blr{\e^{\ji \theta_j}}\hat{s}_j\!\bigg)^{\!H}\!\mb{h}_i\hat{s}_i\e^{\ji \theta_i}\!\bigg\}\!\bigg\} \nonumber  \\
		&\propto p(\theta_i) \exp\bigl(\Re\{\eta_{\mr{t},i}^*\e^{\ji \theta_i}\}\bigr),
	\end{align}
	where $\eta_{\mr{t},i} = 2\hat{\gamma}\mb{h}_i^H\left(\hat{\bs{\Phi}}^H\mb{y} - \sum_{j\neq i}^K \mb{h}_j\blr{\e^{\ji \theta_j}}\hat{s}_j\right)\hat{s}_i^*.$ 
	
	If the prior $p(\theta_i)$ is assumed to be von Mises $\mc{VM}(\varepsilon_{\mr{t},i})$, then the variational density $q(\theta_i)$ is also von Mises with mean $\hat{\theta}_i = \angle(\varepsilon_{\mr{t},i} + \eta_{\mr{t},i})$ and concentration $\kappa_{\mr{t},i} = |\varepsilon_{\mr{t},i} + \eta_{\mr{t},i}|$. Thus, we have $\lr{\e^{\pm \ji \theta_i}} = d_{\mr{t},i} \e^{\pm\ji \hat{\theta}_i}$ , where $d_{\mr{t},i} = {I_1(\kappa_{\mr{t},i})}/{I_0(\kappa_{\mr{t},i})}$.
	The term $\eta_{\mr{t},i}$ can be simplified as
	\begin{align}
		\eta_{\mr{t},i}  
		&= 2\hat{\gamma} \mb{h}_i^H\left(\mb{r} +  \mb{h}_i\blr{\e^{\ji \theta_i}}\hat{s}_i\right)\hat{s}_i^* \nonumber\\
		&=  2\hat{\gamma} \big(\|\mb{h}_i\|^2|\hat{s}_i|^2 d_{\mr{t},i} \e^{\ji \hat{\theta}_i} + \mb{h}_i^H\mb{r}\hat{s}_i^* \big).
	\end{align}

	\emph{3) Update $s_i$:} Taking the expectation of the conditional in \eqref{eq:con:gau:mf} w.r.t. all latent variables except $s_i$, the variational distribution $q(s_i)$ can be derived as \eqref{eq:q-x-vonMises},
	\begin{figure*} 
		\begin{align} \label{eq:q-x-vonMises}
			q(s_i) 
			&\propto p(s_i) 
			\exp\bigg\{-\hat{\gamma}\bblr{\bigg\|\bs{\Phi}^H\mb{y} - \sum_{j\neq i}^K \mb{h}_j\e^{\ji {\theta}_j}s_j - \mb{h}_i\e^{\ji {\theta}_i}s_i \bigg\|^2}_{-s_i}\bigg\} \nonumber \\
			&\propto p(s_i) \exp\bigg\{-\hat{\gamma}\bigg[\|\mb{h}_i\|^2 |s_i|^2 -  2\Re\bigg\{\bigg(\hat{\bs{\Phi}}^H\mb{y} - \sum_{j\neq i}^K  \mb{h}_j\blr{\e^{\ji {\theta}_i}}\hat{s}_j\bigg)^H\mb{h}_i\blr{\e^{\ji {\theta}_i}}s_i\bigg\}\bigg]\bigg\} \nonumber\\
			&\propto p(s_i) \exp\big\{-\hat{\gamma}\|\mb{h}_i\|^2|s_i-z_i|^2\big\},
		\end{align} 
		\hrule 
	\end{figure*} 
	where
    \vspace{-0.1cm}
	\begin{align} 
		z_i &= \frac{\mb{h}_i^H}{\|\mb{h}_i\|^2}\Bigg(\hat{\bs{\Phi}}^H\mb{y} - \sum_{j\neq i}^K \mb{h}_j\blr{\e^{\ji \theta_i}}\hat{s}_j\Bigg)\blr{\e^{-\ji \theta_i}}\nonumber \\
		&= \frac{\mb{h}_i^H}{\|\mb{h}_i\|^2} \left(\mb{r} +  \mb{h}_i\blr{\e^{\ji \theta_i}}\hat{s}_i\right)\blr{\e^{-\ji \theta_i}} \nonumber \\
		&= d_{\mr{t},i}^2\hat{s}_i + d_{\mr{t},i} \e^{-\ji \hat{\theta}_i}\frac{\mb{h}_i^H\mb{r}}{\|\mb{h}_i\|^2} .
	\end{align} 
	
	The variational mean $\hat{s}_i$ and variance $\tau_{s_i}$ are calculated as 
	$\hat{s}_i = \ms{F}_{s}\big(z_i,\hat{\gamma}\|\mb{h}_i\|^2\big) \triangleq \mathbb{E}\big[s_i\!\mid\! z_i = s_i \!+\! \mc{CN}\big(0,1/(\hat{\gamma} \|\mb{h}_i\|^2)\big)\big]$
    and $
    \tau_{s_i} = \ms{G}_{s}\big(z_i,\hat{\gamma}\|\mb{h}_i\|^2\big) \triangleq \mr{Var}\big[s_i\!\mid\! z_i = s_i \!+\! \mc{CN}\big(0,1/(\hat{\gamma} \|\mb{h}_i\|^2)\big)\big]$, respectively, where ${\sf F}_s(\cdot)$ and ${\sf G}_s(\cdot)$ are given in the Appendix B of \cite{nguyen2024variational}.
	
	\emph{4) Update $\gamma$:} Taking the expectation of the conditional in \eqref{eq:con:gau:mf} w.r.t. all latent variables except $\gamma$, the variational distribution $q(\gamma)$ can be obtained as
	\begin{align}
		q(\gamma) &\propto p(\gamma)\exp\left\{\ln p(\mb{y}\!\mid\!\mb{s},\bs{\phi},\bs{\theta},\gamma)\right\} \nonumber\\
		&\propto \exp\big\{M\ln\gamma - \gamma\blr{\|\bs{\Phi}^H\mb{y} - \mb{H}\bs{\Theta}\mb{s}\|^2 \nonumber\\
			&\quad+ (a_\gamma - 1)\ln \gamma - b_\gamma  \gamma}\big\}.
	\end{align}
	The variational distribution $q(\gamma)$ is thus Gamma with mean
	\begin{align} \label{eq:gamma:VM}
		\hat{\gamma} = \frac{ M + a_\gamma}{b_\gamma + \blr{\|\bs{\Phi}^H\mb{y} - \mb{H}\bs{\Theta}\mb{s}\|^2} }.
	\end{align} 
	Since the von Mises prior is used for $\theta_i$, $\blr{\|\bs{\Phi}^H\mb{y} - \mb{H}\bs{\Theta}\mb{s}\|^2}$ can be further derived as
	\begin{align}
		&\blr{\|\bs{\Phi}^H\mb{y} - \mb{H}\bs{\Theta}\mb{s}\|^2}   \nonumber\\
		&= \|\mb{r}\|^2 + \sum_{m=1}^M \mr{Var}\big[\e^{\ji \theta_m}y_m\big] + \sum_{i=1}^K\|\mb{h}_i\|^2\mr{Var}[\e^{\ji \mr{\theta}_i}s_i]\nonumber\\
		&= \|\mb{r}\|^2 + \sum_{m=1}^N|y_m|^2(1-d_{\mr{r},m}^2)  \nonumber \\
		&\quad+ \sum_{i=1}^K \|\mb{h}_i\|^2 \big(\tau_{s_i}  + (1-d_{\mr{t},i}^2)|\hat{s}_i|^2 \big).
	\end{align}
	
	The MF-VB algorithm with von Mises priors for $\theta_i$'s and $\phi_m$'s is summarized in Algorithm~\ref{algo-2}. The iteration index is only included in the superscripts of the estimated variables if the old (iteration $\ell$) and the new (iteration $\ell+1$) estimated values are needed later in the algorithm. In the VB framework, the estimation of a latent variable requires the latest updates of the other latent variables. These updates are reflected in the residual vector $\mb{r}$. Instead of updating the whole vector $\mb{r}$, which induces high complexity, we only update the contribution of the latent variable that has been most recently updated.
	\begin{algorithm}[!t]
		\small
		\caption{-- \textit{\textbf{MF-VB algorithm}} with postulated noise variance and von Mises prior}
		\label{algo-2}
		\begin{algorithmic}[1]
			\State \textbf{Input:} $\mb{y}$, $\mb{H}$, and priors $\big\{p(s_i)\big\}$, $\big\{p(\theta_i)\big\}$, $\big\{p(\phi_m)\big\}$\;
			\State \textbf{Output:} $\hat{\mb{s}}$, $\hat{\bs{\theta}}$, and $\hat{\bs{\phi}}$\;
			\State Initialize $\hat{s}_i^1 = 0$, $\tau_{s_i} = \mr{Var}_{p(s_i)}[s_i]$, $\forall i$; and $\mb{r} = \mb{y}$
			\State Initialize $\hat{\theta}_i^1 = 0$ and  $d_{\mr{t},i}^1 = {I_1(|\varepsilon_{\mr{t},i}|)}/{I_0(|\varepsilon_{\mr{t},i}|)}$,~$\forall i$\;
			\State Initialize $\hat{\phi}_m^1 = 0$, 
			and $d_{\mr{r},m}^1 = {I_1( |\varepsilon_{\mr{r},m}|)}/{I_0(|\varepsilon_{\mr{r},m}|)}$,~$\forall m$
			\For{$\ell=1,2,\ldots,T$}
			\State $\hat{\gamma} \leftarrow M/\big(\|\mb{r}\|^2 + \sum_{m=1}^M|y_m|^2\big(1-(d_{\mr{r},m}^\ell)^2\big)$
			\Statex \quad\quad\quad\quad\quad\quad $+ \sum_{i=1}^K \|\mb{h}_i\|^2 \left[ \tau_{s_i}^\ell +|\hat{s}_i^\ell|^2 (1- (d_{\mr{t},i}^\ell)^2) \right]\big)$
			\For{$i=1,2,\ldots,K$}
			\State $\displaystyle z_i \gets (d_{\mr{t},i}^\ell)^2 \hat{s}_i^\ell +  d_{\mr{t},i}^\ell \e^{-\ji \hat{\theta}_i^\ell} {\mb{h}_i^H\mb{r}}/{\|\mb{h}_i\|^2}$
			\State $\phantom{z_i}\mathllap{\hat{s}_i^{\ell+1}} \gets \mathbb{E}\big[s_i\!\mid\! z_i = s_i + \mc{CN}\big(0,1/(\hat{\gamma} \|\mb{h}_i\|^2)\big)\big]$\; 
			\State $\phantom{z_i}\mathllap{\tau_{s_i}^{\ell+1}} \gets \mr{Var}\big[s_i\!\mid\! z_i = s_i + \mc{CN}\big(0,1/(\hat{\gamma} \|\mb{h}_i\|^2)\big)\big]$\;
			\State $\phantom{z_i}\mathllap{\mb{r}} \gets \mb{r} + \mb{h}_i d_{\mr{t},i}^\ell\e^{\ji \hat{\theta}_i^\ell}(\hat{s}_i^\ell - \hat{s}_i^{\ell+1})$
			\State $\displaystyle\phantom{z_i}\mathllap{\eta_{\mr{t},i}} \gets  2\hat{\gamma} \big(\|\mb{h}_i\|^2|\hat{s}_i^{\ell+1}|^2 d_{\mr{t},i}^\ell \e^{\ji \hat{\theta}_i^\ell} + \mb{h}_i^H\mb{r}(\hat{s}_i^{\ell+1})^* \big)$ 
			\State $\phantom{z_i}\mathllap{\hat{\theta}_i^{\ell+1}} \gets \angle(\varepsilon_{\mr{t},i} + \eta_{\mr{t},i})$
			\State  $\phantom{z_i}\mathllap{\kappa_{\mr{t},i}^{\ell+1}} \gets |\varepsilon_{\mr{t},i} + \eta_{\mr{t},i}|$
			\State $\phantom{z_i}\mathllap{d_{\mr{t},i}^{\ell+1}} \gets {I_1(\kappa_{\mr{t},i}^{\ell+1})}/{I_0(\kappa_{\mr{t},i}^{\ell+1})}$
			\State $\phantom{z_i}\mathllap{\mb{r}} \gets \mb{r} + \mb{h}_i \hat{s}_i^{\ell+1}\big( d_{\mr{t},i}^{\ell}\e^{\ji \hat{\theta}_i^{\ell}} - d_{\mr{t},i}^{\ell+1}\e^{\ji \hat{\theta}_i^{\ell+1}}\big)$
		\EndFor
		\State $\bs{\eta}_{\mr{r}} \gets 2\hat{\gamma} \big(\mb{y}\odot\mb{y}^*\odot \mb{d}_{\mr{r}}^\ell\odot \e^{\ji\hat{\bs{\phi}}^\ell} - \mb{y}\odot\mb{r}^*\big)$
		\State $\phantom{\bs{\eta}_{\mr{r}}} \mathllap{\hat{\bs{\phi}}^{\ell+1}} \gets \angle(\bs{\varepsilon}_{\mr{r}} + \bs{\eta}_{\mr{r}})$ 
		\State $\phantom{\bs{\eta}_{\mr{r}}} \mathllap{\bs{\kappa}_{\mr{r}}^{\ell+1}} \gets |\bs{\varepsilon}_{\mr{r}} + \bs{\eta}_{\mr{r}}|$
		\State $\phantom{\bs{\eta}_{\mr{r}}} \mathllap{\mb{d}_{\mr{r}}^{\ell+1}} \gets {I_1(\bs{\kappa}_{\mr{r}}^{\ell+1})}./{I_0(\bs{\kappa}_{\mr{r}}^{\ell+1})}$
		\State $\phantom{\bs{\nu}_{\mr{r}}} \mathllap{\mb{r}} \gets \mb{r} -  \mb{y}\odot \big( \mb{d}_{\mr{r}}^{\ell}\odot\e^{-\ji \hat{\bs{\phi}}^{\ell}} - \mb{d}_{\mr{r}}^{\ell+1}\odot\e^{-\ji \hat{\bs{\phi}}^{\ell+1}} \big)$\;
		\EndFor
		\State $\forall i: \hat{s}_i \leftarrow \arg \max_{a\in \mathcal{A}}p_a \mathcal{CN}\big({z}_{i}; a,1/(\hat{\gamma}\|\mb{h}_i\|^2)\big)$\;
	\end{algorithmic}
\end{algorithm}

\begin{remark}
	For the case of a single phase noise source $\phi$ at the BS, the variational distribution of $\phi$ at iteration-$\ell$ can be found as
	\begin{align}\label{q-phi}
		q(\phi^{\ell+1}) &\propto p(\phi^{\ell+1})\exp\bigl(\Re\{\eta_{\mr{r}}^*\e^{\ji\phi^{\ell+1}}\}\bigr)\nonumber\\
        &= \mc{VM}\big(\phi^{\ell+1};\varepsilon_{\mr{r}} + \eta_{\mr{r}}\big),
	\end{align}
	where $\eta_{\mr{r}} = 2\hat{\gamma}\,\mb{y}^T \big(\mb{H}\hat{\bs{\Theta}}\hat{\mb{s}})^* = 2\hat{\gamma}\big(\|\mb{y}\|^2d_{\mr{r}}^\ell\e^{\ji\hat{\phi}^\ell} - \mb{r}^H\mb{y}\big)$ with $d_{\mr{r}}^\ell$ and $\hat{\phi}^\ell$ obtained from $q(\phi^\ell)$ at the previous iteration. The update of $\hat{\gamma}$ in Step 7 is also modified by replacing
	$\sum_{m=1}^M |y_m|^2\big(1-(d_{\mr{r},m}^\ell)^2\big)$ with $\big(1-(d_{\mr{r}}^\ell)^2\big)\|\mb{y}\|^2$.
\end{remark}

\subsection{LMMSE-VB Algorithm} \label{sec:LMMSE:von}

The noise precision matrix $\bs{\Gamma} = (\mb{C}^{\rm post})^{-1}$ is floated as a random matrix.
The joint distribution is given by 
\begin{align} \label{eq:con:gau:lmmse}
p(\mb{y},\mb{s},\bs{\theta},\bs{\phi},\bs{\Gamma}) = p(\mb{y}\!\mid\!\mb{s},\bs{\theta},\bs{\phi},\bs{\Gamma}) p(\mb{s})p(\bs{\theta})p(\bs{\phi})p(\bs{\Gamma}),
\end{align}
where $p(\mb{y}\!\mid\!\mb{s},\bs{\phi},\bs{\theta},\bs{\Gamma}) = \mc{CN}\big(\bs{\Phi}^H\mb{y};\mb{H}\bs{\Theta}\mb{s} ,\bs{\Gamma}^{-1}\big)$. We approximate the posterior distribution using a mean-field factorization
\begin{align}
p(\mb{s},\bs{\theta},\bs{\phi},\bs{\Gamma}\!\mid\!\mb{y};\mb{H}) &\approx q(\mb{s},\bs{\theta},\bs{\phi},\bs{\Gamma}) \nonumber \\
&=\Bigg[\prod_{i=1}^K q(s_i)q({\theta}_i)\Bigg] \prod_{m=1}^M q({\phi}_m) q(\bs{\Gamma}).
\end{align}

\emph{1) Update $\phi_m$:} Taking the expectation of the conditional in \eqref{eq:con:gau:lmmse} w.r.t. all latent variables except $\phi_m$, the variational distribution $q(\phi_m)$ can be derived as 
\begin{align}
	&	q(\phi_m) \nonumber \\
	&\propto p(\phi_m) \exp\Big\{\!\!-\!\blr{(\bs{\Phi}^H\mb{y}\! -\! \mb{H}\bs{\Theta}\mb{s})^H\bs{\Gamma}(\bs{\Phi}^H\mb{y}\! -\! \mb{H}\bs{\Theta}\mb{s})}_{-\phi_m}\Big\} \nonumber\\
	&\propto p(\phi_m)\exp\bigg\{2\,\Re\big\{\e^{\ji \phi_m}y_m^*\big[\hat{\bs{\Gamma}}\big]_{mm}\mb{H}_{m,:}\hat{\bs{\Theta}}\hat{\bs{\mb{s}}}\big\} \nonumber \\
	&\quad\quad\quad\quad - 2\! \sum_{n\neq m}^M\! \Re\big\{\e^{\ji \phi_m}y_m^*\big[\hat{\bs{\Gamma}}\big]_{mn}\big(\blr{\e^{-\ji \theta_n}}y_n-\mb{H}_{n,:}\hat{\bs{\Theta}}\hat{\mb{s}}\big)\big\}\!\bigg\}\nonumber\\
	&\propto p(\phi_m) \exp \big\{\Re\big\{\eta_{\mr{r},m}^*\e^{\ji \phi_m}\big\}\big\},
\end{align}
where we define 
\begin{align}
	\eta_{\mr{r},m} &= 2y_m \bigg(\big[\hat{\bs{\Gamma}}\big]_{mm}\mb{H}_{m,:}\hat{\bs{\Theta}}\hat{\bs{\mb{s}}}\nonumber \\
	&\quad \quad \quad -\sum_{n\neq m}^M \big[\hat{\bs{\Gamma}}\big]_{mn}\big(\blr{\e^{-\ji \theta_n}}y_n-\mb{H}_{n,:}\hat{\bs{\Theta}}\hat{\mb{s}}\big)\bigg)^*  \nonumber \\
	&= 2y_m \bigg(\big[\hat{\bs{\Gamma}}\big]_{mm}\blr{\e^{-\ji \phi_m}}y_m \nonumber \\
	&\quad \quad \quad -\sum_{n=1}^M \big[\hat{\bs{\Gamma}}\big]_{mn}\big(\blr{\e^{-\ji \theta_n}}y_n-\mb{H}_{n,:}\hat{\bs{\Theta}}\hat{\mb{s}}\big)\bigg)^*  \nonumber \\
	&= 2y_m \big(\big[\hat{\bs{\Gamma}}\big]_{mm}d_{\mr{r},m} \e^{-\ji \hat{\phi}_m}y_m - [\hat{\bs{\Gamma}}]_{:,m}^H\mb{r}\big)^* \nonumber\\
	&= 2\big[\hat{\bs{\Gamma}}\big]_{mm}|y_m|^2 d_{\mr{r},m} \e^{\ji \hat{\phi}_m} - 2y_m \mb{r}^H[\hat{\bs{\Gamma}}]_{:,m}.
\end{align}

If the prior is von Mises $\mc{VM}(\varepsilon_{\mr{r},m})$, $q(\phi_m)$ is also von Mises with mean $\hat{\phi}_n = \angle(\varepsilon_{\mr{r},m} + \eta_{\mr{r},m})$ and concentration $\kappa_{\mr{r},m} = |\nu_{\mr{r},m} + \eta_{\mr{r},m}|$. 

\emph{2) Update $\theta_i$:} The variational distribution $q(\theta_i)$ can be derived as in \eqref{eq:q:theta:lmmse:VM},
\begin{figure*}
	\begin{align} \label{eq:q:theta:lmmse:VM}
		q(\theta_i) 
		&\propto p(\theta_i) \exp\Big\{-\blr{\big(\bs{\Phi}^H\mb{y} - \mb{H}\bs{\Theta}\mb{s}\big)\bs{\Gamma}\big(\bs{\Phi}^H\mb{y} - \mb{H}\bs{\Theta}\mb{s}\big)}_{-\theta_i}\Big\} \nonumber\\
		&\propto p(\theta_i) \exp\bigg\{-\bblr{\bigg(\bs{\Phi}^H\mb{y} - \sum_{j\neq i}^K \mb{h}_j\e^{\ji \theta_j}s_j - \mb{h}_i\e^{\ji \theta_i}s_i\bigg)^H\bs{\Gamma}\bigg(\bs{\Phi}^H\mb{y} - \sum_{j\neq i}^K \mb{h}_j\e^{\ji \theta_j}s_j - \mb{h}_i\e^{\ji \theta_i}s_i\bigg)}_{-\theta_i}\bigg\} \nonumber\\
		&\propto p(\theta_i) \exp\bigg\{2\Re\bigg\{\bigg(\hat{\bs{\Phi}}^H\mb{y} - \sum_{j\neq i}^K \mb{h}_j\blr{\e^{\ji \theta_j}}\hat{s}_j\bigg)^H\hat{\bs{\Gamma}}\mb{h}_i\hat{s}_i\e^{\ji \theta_i}\bigg\}\bigg\} \nonumber\\
		&\propto p(\theta_i) \exp\big\{\Re\{\eta_{\mr{t},i}^*\e^{\ji \theta_i}\}\big\},
	\end{align}
	\hrule
	\begin{align} \label{eq:q:s_i:lmmse:vm}
		q(s_i) 
		&\propto p(s_i) \exp\Big\{-\blr{\big(\bs{\Phi}^H\mb{y} - \mb{H}\bs{\Theta}\mb{s}\big)\bs{\Gamma}\big(\bs{\Phi}^H\mb{y} - \mb{H}\bs{\Theta}\mb{s}\big)}_{-s_i}\Big\} \nonumber\\
		&\propto p(s_i) \exp\bigg\{-\bblr{\bigg(\bs{\Phi}^H\mb{y} - \sum_{j\neq i}^K \mb{h}_j\e^{\ji \theta_j}s_j - \mb{h}_i\e^{\ji \theta_i}s_i\bigg)^H\bs{\Gamma}\bigg(\bs{\Phi}^H\mb{y} - \sum_{j\neq i}^K \mb{h}_j\e^{\ji \theta_j}s_j - \mb{h}_i\e^{\ji \theta_i}s_i\bigg)}_{-s_i}\bigg\}\nonumber \\
		&\propto p(s_i) \exp\bigg\{-\mb{h}_i^H\hat{\bs{\Gamma}}\mb{h}_i |s_i|^2 + 2\Re\bigg\{\bigg(\hat{\bs{\Phi}}^H\mb{y} - \sum_{j\neq i}^K \mb{h}_j\blr{\e^{\ji \theta_j}}\hat{s}_j\bigg)^H\hat{\bs{\Gamma}}\mb{h}_i\blr{\e^{\ji \theta_i}}s_i\bigg\}\bigg\} \nonumber\\
		&\propto p(s_i)\exp \big\{-\mb{h}_i^H\hat{\bs{\Gamma}}\mb{h}_i \big(|s_i|^2 + 2\Re\{z_i^*s_i\}\big)\big\} \nonumber \\
		&\propto p(s_i)\exp\big\{-\mb{h}_i^H\hat{\bs{\Gamma}}\mb{h}_i|s_i-z_i|^2 \big\},
	\end{align}
	\hrule
\end{figure*} 
where 
\begin{align}
	\eta_{\mr{t},i} = 2\mb{h}_i^H\hat{\bs{\Gamma}}\left(\hat{\bs{\Phi}}^H\mb{y} - \sum_{j\neq i}^K \mb{h}_j\blr{\e^{\ji \theta_j}}\hat{s}_j\right)\hat{s}_i^*.  
\end{align}

If the prior $p(\theta_i)$ is von Mises $\mc{VM}(\varepsilon_{\mr{t},i})$, then the variational density $q(\theta_i)$ is also von Mises with mean $\hat{\theta}_i = \angle(\varepsilon_{\mr{t},i} + \eta_{\mr{t},i})$ and concentration $\kappa_{\mr{t},i} = |\varepsilon_{\mr{t},i} + \eta_{\mr{t},i}|$.
The term $\eta_{\mr{t},i}$ can be simplified as
\begin{align}
	\eta_{\mr{t},i} &= 2\mb{h}_i^H \hat{\bs{\Gamma}}\bigg(\hat{\bs{\Phi}}^H\mb{y} - \sum_{j\neq i}^K \mb{h}_j\blr{\e^{\ji \theta_j}}\hat{s}_j\bigg)\hat{s}_i^* \nonumber \nonumber\\
	&= 2 \mb{h}_i^H\hat{\bs{\Gamma}}\left(\mb{r} +  \mb{h}_i\blr{\e^{\ji \theta_i}}\hat{s}_i\right)\hat{s}_i^* \nonumber\\
	&=  2 \big(\mb{h}_i^H\hat{\bs{\Gamma}}\mb{h}_i|\hat{s}_i|^2 d_{\mr{t},i} \e^{\ji \hat{\theta}_i} + \mb{h}_i^H\hat{\bs{\Gamma}}\mb{r}\hat{s}_i^*\big).
\end{align}

\emph{3) Update $s_i$:} Taking the expectation of the conditional in \eqref{eq:con:gau:lmmse} w.r.t. all latent variables except $s_i$, the variational distribution $q(s_i)$ can be derived as in \eqref{eq:q:s_i:lmmse:vm},
where
\begin{align} 
	z_i &= \frac{\mb{h}_i^H\hat{\bs{\Gamma}}\Big(\hat{\bs{\Phi}}^H\mb{y} - \sum_{j\neq i}^K \mb{h}_j\blr{\e^{\ji \theta_j}}\hat{s}_j\Big)\blr{\e^{-\ji \theta_i}}}{\mb{h}_i^H\hat{\bs{\Gamma}}\mb{h}_i} \nonumber\\
	&= \frac{\mb{h}_i^H\hat{\bs{\Gamma}}\big(\mb{r} + \mb{h}_i\blr{\e^{\ji \theta_i}} \hat{s}_i\big)\blr{\e^{-\ji \theta_i}}}{\mb{h}_i^H\hat{\bs{\Gamma}}\mb{h}_i} \nonumber\\
	&= d_{\mr{t},i}^2\hat{s}_i + d_{\mr{t},i} \e^{-\ji \hat{\theta}_i}\frac{\mb{h}_i^H\hat{\bs{\Gamma}}\mb{r}}{\mb{h}_i^H\hat{\bs{\Gamma}}\mb{h}_i}.
\end{align}

\emph{4) Update $\bs{\Gamma}$:} The variational distribution $q(\bs{\Gamma})$ can be derived as
\begin{align}
	q(\bs{\Gamma}) &\propto p(\bs{\Gamma})\exp\left\{\ln p(\mb{y}\!\mid\!\mb{s},\bs{\theta},\bs{\phi},\bs{\Gamma})\right\} \nonumber \\
	&\propto p(\bs{\Gamma})\exp \big\{\ln|\bs{\Gamma}| \nonumber\\
	&\quad - \blr{(\bs{\Phi}^H\mb{y} - \mb{H}\bs{\Theta}\mb{s})^H\bs{\Gamma}(\bs{\Phi}^H\mb{y} - \mb{H}\bs{\Theta}\mb{s})}\big\}.
\end{align}
Applying Lemma 1 in \cite{nguyen2024variational} with deterministic of $\mb{H}$, we have
\begin{align}
	&\blr{(\bs{\Phi}^H\mb{y} - \mb{H}\bs{\Theta}\mb{s})^H\bs{\Gamma}(\bs{\Phi}^H\mb{y} - \mb{H}\bs{\Theta}\mb{s})}\nonumber\\
	&= (\hat{\bs{\Phi}}^H\mb{y} - \mb{H}\hat{\bs{\Theta}}\hat{\mb{s}})^H\bs{\Gamma}(\hat{\bs{\Phi}}^H\mb{y} - \mb{H}\hat{\bs{\Theta}}\hat{\mb{s}}) \nonumber \\
	&\quad + \tr\{\bs{\Sigma}_{\bs{\Phi}^H\mb{y}}\bs{\Gamma}\}  +  \tr\{ \mb{H}\bs{\Sigma}_{\bs{\Theta}\mb{s}}\mb{H}^H\bs{\Gamma}\} \nonumber \\
	& = \tr\big\{\big(\mb{r}\mb{r}^H + \bs{\Sigma}_{\bs{\Phi}^H\mb{y}} + \mb{H}\bs{\Sigma}_{\bs{\Theta}\mb{s}}\mb{H}^H\big)\bs{\Gamma}\big\},
\end{align}
where
\begin{align}
	\bs{\Sigma}_{\bs{\Phi}^H\mb{y}} &= \mr{diag}\big(\var[\e^{-\ji\phi_1}y_1], \ldots, \var[\e^{-\ji\phi_M}y_M]\big) \nonumber \\
	&= \mr{diag}\big(|y_1|^2(1-d_{\mr{r},1}^2),\ldots, |y_M|^2(1-d_{\mr{r},M}^2)\big),
\end{align}
and 
\begin{align}
	\bs{\Sigma}_{\bs{\Theta}\mb{s}} &= \mr{diag}\big(\var[\e^{\ji\theta_1}s_1], \ldots, \var[\e^{\ji\theta_K}s_K]\big) \nonumber \\ 
	&= \mr{diag}\big(\tau_{s_1}\!+\!(1-d_{\mr{t},1}^2)|\hat{s}_1|^2,\ldots, \tau_{s_K}\!+\!(1-d_{\mr{t},K}^2)|\hat{s}_K|^2\big).
\end{align}
Due to the rank deficiency of $\blr{(\bs{\Phi}^H\mb{y} - \mb{H}\bs{\Theta}\mb{s})^H\bs{\Gamma}(\bs{\Phi}^H\mb{y} - \mb{H}\bs{\Theta}\mb{s})}$, and following \cite{nguyen2024variational}, we propose the estimator for $\hat{\bs{\Gamma}}$ as 
\begin{align} \label{eq:Gamma:Gau}
	\hat{\bs{\Gamma}} = \left(\frac{\|\mb{r}\|^2}{M}\mb{I}_M + \bs{\Sigma}_{\bs{\Phi}^H\mb{y}} + \mb{H}\bs{\Sigma}_{\bs{\Theta}\mb{s}}\mb{H}^H \right)^{-1}.
\end{align}
The LMMSE-VB process with von Mises prior is summarized in Algorithm~\ref{algo-3}.

\begin{remark}
	For the case of one phase noise source $\phi$ at the BS, the variational distribution of $\phi$ at iteration-$\ell$ can be found as in \eqref{q-phi},
	where $\eta_{\mr{r}} = 2\mb{y}^T\hat{\bs{\Gamma}} \big(\mb{H}\hat{\bs{\Theta}}\hat{\mb{s}})^* = 2\big( \mb{y}^H\hat{\bs{\Gamma}}\mb{y}d_{\mr{r}}^\ell\e^{\ji\hat{\phi}^\ell} - \mb{r}^H\hat{\bs{\Gamma}}\mb{y}\big)$ with 
	$d_{\mr{r}}^\ell$ and $\hat{\phi}^\ell$ obtained from $q(\phi^\ell)$ at the previous iteration. The update of $\hat{\bs{\Gamma}}$ at iteration-$\ell$ is also modified by assigning
	$\hat{\bs{\Sigma}}^{\ell}_{\bs{\Phi}^H\mb{y}} = (1-(d_{\mr{r}}^\ell)^2)\mb{y}\mb{y}^H$.
\end{remark}

\begin{algorithm}[!th]
	\small
	\caption{-- \textit{\textbf{LMMSE-VB algorithm}} with postulated noise covariance and von Mises prior}
	\label{algo-3}
	\begin{algorithmic}[1]
		\State \textbf{Input:} $\mb{y}$, $\mb{H}$, and priors $\big\{p(s_i)\big\}$, $\big\{p(\theta_i)\big\}$, $\big\{p(\phi_m)\big\}$\;
		\State \textbf{Output:} $\hat{\mb{s}}$, $\hat{\bs{\theta}}$, and $\hat{\bs{\phi}}$\;
		\State Initialize $\hat{s}_i^1 = 0$, $\tau_{s_i} = \mr{Var}_{p(s_i)}[s_i]$, $\forall i$; and $\mb{r} = \mb{y}$
		\State Initialize $\hat{\theta}_i^1 = 0$ and  $d_{\mr{t},i}^1 = {I_1(|\varepsilon_{\mr{t},i}|)}/{I_0(|\varepsilon_{\mr{t},i}|)}$,~$\forall i$\;
		\State Initialize $\hat{\phi}_m^1 = 0$, 
		and $d_{\mr{r},m}^1 = {I_1( |\varepsilon_{\mr{r},m}|)}/{I_0(|\varepsilon_{\mr{r},m}|)}$,~$\forall m$
		\For{$\ell=1,2,\ldots$}
		\State $\displaystyle \bs{\Sigma}_{\bs{\Phi}^H\mb{y}}^\ell \gets \mr{diag}\big(|y_1|^2(1-(d_{\mr{r},1}^{\ell})^2),\ldots, |y_M|^2(1-(d_{\mr{r},M}^\ell)^2)\big)$
		\State $\phantom{\bs{\Sigma}_{\bs{\Phi}^H\mb{y}}^\ell} \mathllap{\bs{\Sigma}_{\bs{\Theta}\mb{s}}^\ell} \gets \mr{diag}\big( \tau_{s_1}^\ell + |\hat{s}_1^\ell|^2(1-(d_{\mr{t},1}^\ell)^2),\ldots, \tau_{s_K}^\ell +  |\hat{s}_K^\ell|^2(1-(d_{\mr{t},K}^\ell)^2)\big)$
		\State $\displaystyle \phantom{\bs{\Sigma}_{\bs{\Phi}^H\mb{y}}^\ell} \mathllap{\hat{\bs{\Gamma}}} \gets \left(\frac{\|\mb{r}\|^2}{M}\mb{I}_M + \bs{\Sigma}_{\bs{\Phi}^H\mb{y}}^\ell + \mb{H}\bs{\Sigma}_{\bs{\Theta}\mb{s}}^\ell\mb{H}^H\right)^{-1}$
		\For{$i=1,2,\ldots,K$}
			\State $\displaystyle z_i \gets (d_{\mr{t},i}^\ell)^2 \hat{s}_i +  d_{\mr{t},i}^\ell \e^{-\ji \hat{\theta}_i} {\mb{h}_i^H\hat{\bs{\Gamma}}\mb{r}}/{\mb{h}_i^H\hat{\bs{\Gamma}}\mb{h}_i}$
			\State $\phantom{z_i}\mathllap{\hat{s}_i^{\ell+1}} \gets \mathbb{E}\big[s_i \!\mid\! z_i = s_i + \mc{CN}\big(0,1/\mb{h}_i^H\hat{\bs{\Gamma}}\mb{h}_i\big)\big]$\; 
			\State $\phantom{z_i}\mathllap{\tau_{s_i}^{\ell+1}} \gets \mr{Var}\big[s_i \!\mid\! z_i = s_i + \mc{CN}\big(0,1/\mb{h}_i^H\hat{\bs{\Gamma}}\mb{h}_i\big)\big]$\; 
			\State $\phantom{z_i}\mathllap{\mb{r} }\gets \mb{r} + \mb{h}_i d_{\mr{t},i}^\ell \e^{\ji \hat{\theta}_i^\ell}(\hat{s}_i^\ell - \hat{s}_i^{\ell+1})$
			\State $\displaystyle \phantom{z_i}\mathllap{\eta_{\mr{t},i}} \gets  2 \big(\mb{h}_i^H\hat{\bs{\Gamma}}\mb{h}_i|\hat{s}_i^{\ell+1}|^2 d_{\mr{t},i}^\ell \e^{\ji \hat{\theta}_i^\ell} + \mb{h}_i^H\hat{\bs{\Gamma}}\mb{r}(\hat{s}_i^{\ell+1})^*\big)$ 
			\State $\phantom{z_i}\mathllap{\hat{\theta}_i^{\ell+1}} \gets \angle(\varepsilon_{\mr{t},i} + \eta_{\mr{t},i})$ 
			\State $\phantom{z_i}\mathllap{\kappa_{\mr{t},i}^{\ell+1}} \gets |\varepsilon_{\mr{t},i} + \eta_{\mr{t},i}|$
				\State $\phantom{z_i}\mathllap{d_{\mr{t},i}^{\ell+1}} \gets {I_1(\kappa_{\mr{t},i}^{\ell+1})}/{I_0(\kappa_{\mr{t},i}^{\ell+1})}$
			\State $\phantom{z_i}\mathllap{\mb{r}} \gets \mb{r} + \mb{h}_i \hat{s}_i^{t+1}\big( d_{\mr{t},i}^{\ell}\e^{\ji \hat{\theta}_i^{\ell}} - d_{\mr{t},i}^{\ell+1}\e^{\ji \hat{\theta}_i^{\ell+1}}\big)$
		\EndFor
		\For{$m=1,2,\ldots,M$}
		\State $\eta_{\mr{r},m} \gets  2\big[\hat{\bs{\Gamma}}\big]_{mm}|y_m|^2 d_{\mr{r},m}^\ell \e^{\ji \hat{\phi}_m^\ell } - 2y_m \mb{r}^H[\hat{\bs{\Gamma}}]_{:,m}$
		\State $\phantom{\eta_{\mr{r},m}}\mathllap{\hat{\phi}_n^{\ell+1}} \gets \angle(\varepsilon_{\mr{r},m} + \eta_{\mr{r},m})$
		\State $\phantom{\eta_{\mr{r},m}}\mathllap{\kappa_{\mr{r},m}^{\ell+1}} \gets |\varepsilon_{\mr{r},m} + \eta_{\mr{r},m}|$
		\State $\phantom{\eta_{\mr{r},m}} \mathllap{d_{\mr{r},m}^{\ell+1}} \gets {I_1(\kappa_{\mr{r},m}^{\ell+1})}/{I_0(\kappa_{\mr{r},m}^{\ell+1})}$
		\State $\phantom{\nu_{\mr{r},m}}\mathllap{{r}_m} \gets {r}_m - y_m\big( d_{\mr{r},m}^{\ell}\e^{-\ji \hat{{\phi}}_m^{\ell}} - d_{\mr{r},m}^{\ell+1}\e^{-\ji \hat{{\phi}}_m^{\ell+1}} \big)$
		\EndFor
		\EndFor
		\State $\forall i: \hat{s}_i \leftarrow \arg \max_{a\in \mathcal{A}}p_a \mathcal{CN}\big(z_{i}; a,1/\mb{h}_i^H\hat{\bs\Gamma}\mb{h}_i\big)$\;
	\end{algorithmic}
\end{algorithm}

\section{Variational Bayes for Inferring $\mb{x}$} \label{sec:ImprovedMF}
In Section \ref{sec:VBforMIMO}, VB algorithms have been developed to jointly estimate the transmitted symbols $\mb{s}$ and the phase noises $\bs{\theta}$ and $\bs{\phi}$. However, the multiplicative coupling between $\mb{s}$ and $\bs{\theta}$ in the posterior updates introduces nontrivial dependencies that may reduce the detection performance.
To address this issue, we propose an alternative formulation in which the transmitter phase noise is absorbed directly into the transmitted signal, defining the composite variable $x_i$ as shown in \eqref{eq:x_i}. Rather than estimating $s_i$ and $\theta_i$  separately, the VB framework is applied directly to infer $x_i$. As established in Theorem 1 and Corollary 1, this reformulation admits a tractable closed-form posterior, enabling efficient variational inference with reduced posterior mismatch and improved detection performance.

\begin{theorem}[Posteriors of $s$, $\theta$, and $x$ given $z$]
\label{thm:main}
Let ${s \in \mc{S}=\{s_1,\ldots,s_N\}\subset\mathbb{C}}$ be a discrete RV with prior probabilities $p_k = P(s = s_k)$.  Let $\theta \sim \mc{VM}(\varepsilon)$ be a phase noise RV with PDF $p(\theta;\varepsilon) \propto \exp\bigl(\Re\{\varepsilon^*\e^{\ji\theta}\}\bigr)$. 
The variables $s$ and $\theta$ are mutually independent. Let $z$ be a noisy observation of $x\triangleq s\,\e^{\ji\theta}$, i.e.,
\begin{equation}
	z = x + n = s\,e^{\ji\theta} + n,
	\quad n \sim \mathcal{CN}(0,\,\gamma^{-1}),
	\label{eq:model}
\end{equation}
where $n$ is independent of $s$ and $\theta$. For each $k \in \{1,\ldots,N\}$, define the \emph{posterior natural parameter}
\begin{equation}
\eta_k \triangleq \varepsilon + 2\gamma\,z\,s_k^*.	\label{eq:eta}
\end{equation}
The following results hold:

\noindent (I) \emph{Signal prior $p(x)$:} 
\begin{align}
	p(x) = \sum_{k=1}^N p_k\cdot \frac{\exp\big(\Re\{\varepsilon^* \e^{\ji (\angle x - \angle s_k)}\}\big)}{2\pi|s_k|I_0(|\varepsilon|)} \cdot \delta(|x| - |s_k|).
	\tag{I}
	\label{eq:x}
\end{align}
\noindent (II) \emph{Symbol posterior $p(s_k \!\mid\! z)$:}
\begin{equation}
\alpha_k \triangleq P(s = s_k \!\mid\! z) \propto p_k
		\cdot\exp\bigl(-\gamma\,|s_k|^2\bigr)
		\cdot I_0(|\eta_k|),
	\tag{II}
	\label{eq:ps}
\end{equation}
normalized so that $\sum_{k=1}^N \alpha_k = 1$.

\noindent (III) \emph{Phase posterior $p(\theta \!\mid\! z)$:}
\begin{equation}
	p(\theta \!\mid\! z) = \sum_{k=1}^{N} \alpha_k\cdot \mc{VM}(\theta;\eta_k), 
		\quad \theta\in[-\pi,\pi),
	\tag{III}
	\label{eq:ptheta}
\end{equation}
a finite mixture of von Mises distributions with a conditional mean 
\begin{equation}
	\mathbb{E}[\theta\!\mid\! z] = \scalebox{1.5}{$\angle$} \left(\sum_{k=1}^N \alpha_k \cdot \frac{I_1(|\eta_k|)}{I_0(|\eta_k|)} \cdot \e^{\ji\angle \eta_k}\right).
\end{equation}

\noindent(IV) \emph{Signal posterior $p(x \!\mid\! z)$:}
\begin{equation}
{p(x \!\mid\! z)
		= \sum_{k=1}^{N} \alpha_k \cdot
		\frac{\exp\bigl(\Re\bigl\{\eta_k^*\e^{\ji(\angle x - \angle s_k)}\bigr\}\bigr)}
		{2\pi\,|s_k|\,I_0(|\eta_k|)}
		\cdot \delta(|x| - |s_k|),}
	\tag{IV}
	\label{eq:px}
\end{equation}
a finite mixture of von Mises distributions on concentric circles of radii $\{|s_k|\}$.
\end{theorem}

\begin{IEEEproof} The proof proceeds in five steps.\\
\noindent\textit{Step 1: Prior of $x$.} Conditioned on $s=s_k$, $|x| = |s_k|$ and its phase $\angle x = \angle s_k + \theta$. Thus, $\angle x$ is von Mises distributed with PDF $p(\angle x \!\mid\! s_k) = \mc{VM}(\angle x;\varepsilon\,\e^{\ji \angle s_k})$. Hence, the PDF $p(x \!\mid\! s_k)$ is given by
\begin{align}
	p(x\!\mid\! s_k) = \frac{\exp\big(\Re\{\varepsilon^* \e^{\ji (\angle x - \angle s_k)} \}\big)}{2\pi|s_k|I_0(|\varepsilon|)} \cdot \delta(|x| - |s_k|).
\end{align}
By the law of total probability, the marginal $p(x) = \sum_{k=1}^N p(x\!\mid\! s_k)\, p(s_k)$ yields \eqref{eq:x}.

\noindent\textit{Step 2: Conjugate posterior of $\theta$ given $z$ and $s_k$.} Since $n\sim\mc{CN}(0,\gamma^{-1})$, we have
\begin{equation}
	p(z \!\mid\! \theta, s_k) \propto \exp\big(-\gamma|z - s_k \e^{\ji\theta}|^2\big).
\end{equation}
Combining with the von Mises prior, $p(\theta)$, with $\theta$ and $s$ being independent, we have 
\begin{align} \label{eq:eta-PDF}
	p(\theta \!\mid\! z, s_k) &\propto p(z\!\mid\!\theta,s_k)p(\theta) \nonumber 	\\
	&\propto \exp\big(-\gamma|z - s_k \e^{\ji\theta}|^2\big)  \exp\big(\Re\bigl\{\varepsilon^*\e^{\ji\theta}\big\}\bigr) \nonumber \\
	&\propto \exp\big(2\gamma \Re\big\{z^*s_k\e^{\ji \theta}\big\}\big)\exp\big(\Re\bigl\{\varepsilon^*\e^{\ji\theta}\big\}\bigr)  \nonumber \\
	&\propto \exp\big(\Re\big\{\eta_k^* \e^{\ji\theta}\big\}\big),
\end{align}
which is $\mc{VM}(\theta;\eta_k)$, where $\eta_k$ is the posterior natural parameter defined in \eqref{eq:eta}.

\medskip\noindent\textit{Step 3: Marginal likelihood and symbol posterior.}
Using $\int_{-\pi}^{\pi} \e^{\Re\{\eta^* \e^{\ji\theta}\}}\,\mr{d}\theta = 2\pi\,I_0(|\eta|)$:
\begin{align}
	&p(z \!\mid\! s_k) \nonumber \\
	&= \int_{-\pi}^{\pi} p(z \!\mid\! \theta,s_k)\,p(\theta)\,\mr{d}\theta \nonumber \\
	& \propto \int_{-\pi}^{\pi}\! \exp\big(\!-\!\gamma\big[|s_k|^2\! -2\Re\big\{z^*s_k\e^{\ji\theta}\big\}\big]\big) \exp\big(\Re\bigl\{\varepsilon^*\e^{\ji\theta}\big\}\bigr) \,\mr{d}\theta\nonumber \\
	&\propto \e^{-\gamma|s_k|^2} \int_{-\pi}^{\pi} \exp\big(\Re\big\{\eta_k^* \e^{\ji\theta}\big\}\big) \,\mr{d}\theta \nonumber\\
	&\propto \e^{-\gamma|s_k|^2} I_0(|\eta_k|).
\end{align}
Applying Bayes' rule $p(s_k\!\mid\! z) \propto p(z\!\mid\! s_k) p(s_k)$ with prior $p_k$ and normalizing then gives \eqref{eq:ps}.

\medskip\noindent\textit{Step 4: Phase posterior.}
By the law of total probability:
\begin{equation}
	p(\theta \!\mid\! z) = \sum_{k=1}^N p(\theta \!\mid\! z, s_k)\,p(s_k\!\mid\!z)
	= \sum_{k=1}^N \alpha_k \,\mc{VM}(\theta;\eta_k),
\end{equation}
giving \eqref{eq:ptheta}.

\medskip\noindent\textit{Step 5: Signal posterior.} Since $p(\theta \!\mid\! z, s_k)$ is obtained as von Mises as given in \eqref{eq:eta-PDF} and $x = s\,\e^{\ji \theta}$, it follows that given $z$ and $s=s_k$, one has $|x| = |s_k|$ and $\angle x$ as von Mises distributed with the PDF $\mc{VM}(\angle x; \eta_k\, \e^{\ji\angle s_k})$. Thus, the PDF $p(x|z,s_k)$ is given by
\begin{equation}
	p(x \!\mid\! z,s_k) = 	\frac{\exp\bigl(\Re\bigl\{\eta_k^*\e^{\ji(\angle x - \angle s_k)}\bigr\}\bigr)}
	{2\pi\,|s_k|\,I_0(|\eta_k|)}
	\cdot \delta(|x| - |s_k|),
\end{equation}
By the law of total probability, ${p(x\!\mid\! z) = \sum_{k=1}^N p(x\!\mid\! z, s_k) \, p(s_k\!\mid\! z)}$ gives \eqref{eq:px}.
\end{IEEEproof}

\begin{remark}
	The posterior natural parameter $\eta_k = \varepsilon + 2\gamma\,z\,s_k^*$ decomposes
	cleanly: the prior $\varepsilon$ encodes baseline phase uncertainty with mean
	$\angle \varepsilon$ and concentration $|\varepsilon|$, while $2\gamma\,z\,s_k^*$ is the
	matched filter output --- the correlation of observation $z$ against template $s_k$,
	scaled by twice the noise precision.  Since $\angle(z\,s_k^*) = \angle z - \angle s_k$,
	the likelihood steers the posterior mean toward $\angle z - \angle s_k$ as $\gamma$
	increases, which is the maximum-likelihood phase estimate for symbol $s_k$.
	The posterior concentration $|\eta_k| \geq |\varepsilon|$ always exceeds the prior,
	confirming that the observation can only sharpen the phase estimate.
\end{remark}

\begin{corollary}[Posterior Mean and Variance of $x$]
	\label{cor:meanvar}
	Let $d_k \triangleq I_1(|\eta_k|)/I_0(|\eta_k|)$ be the \emph{mean resultant length} of the $k$-th component, and define the per-component posterior mean
	\begin{equation}
	 \mathbb{E}[x \!\mid\! z,\, s_k] = s_k\,d_k\,\e^{\ji\angle\eta_k}.
		\label{eq:xhat}
	\end{equation}
	
	\medskip\noindent (i) \emph{Posterior mean:}
	\begin{equation}
	\bar{x} =\mathbb{E}[x \!\mid\! z]
			= \sum_{k=1}^N \alpha_k\,s_k\,d_k\,\e^{\ji\angle\eta_k}.
		\label{eq:mean}
	\end{equation}
	Each symbol contributes its amplitude $|s_k|$, shrunk by the mean resultant length
	$d_k \in [0,1)$ and rotated to the posterior mean phase $\angle s_k + \angle\eta_k$.
	
	\medskip\noindent (ii) \emph{Posterior variance:}
	\begin{align}
		\mathrm{Var}[x \!\mid\! z]= \sum_{k=1}^N \alpha_k|s_k|^2 - |\bar{x}|^2.
		\label{eq:var}
	\end{align}
\end{corollary}
\begin{IEEEproof}
	Since $\angle x \!\mid\! z, s_k \sim \mc{VM}(\eta_k \e^{\ji\angle s_k})$, one has
	\begin{align}
	\mathbb{E}[x \!\mid\! z,\, s_k]
		&= |s_k|\,\mathbb{E}\!\left[\e^{\ji\angle x}\!\mid\! z,s_k\right]\nonumber \\
		&= |s_k|\,d_k\,\e^{\ji(\angle \eta_k+\angle s_k)}
		= s_k\,d_k\,\e^{\ji \angle\eta_k},
	\end{align}
	giving \eqref{eq:xhat} and \eqref{eq:mean} by linearity.
%
%
The variance is given by
\begin{align}
\mr{Var}[x\!\mid\! z] &= \mathbb{E}\big[|x|^2 \!\mid\! z] - \big|\mathbb{E}[x\!\mid\! z]\big|^2 \nonumber\\
&= \sum_{k=1}^N \mathbb{E}[|x|^2\!\mid\! z,s_k] \,p(s_k\!\mid\! z)- |\bar{x}|^2.
\end{align}
Since $|x|^2 = |s_k|^2$ given $s = s_k$, this gives the results in \eqref{eq:var}.
\end{IEEEproof}

\subsection{Improved MF-VB Algorithm} 
	We assume a Gamma distribution $p(\gamma) = {\rm Gamma}(\alpha,\beta)$, as a conjugate prior for the noise precision $\gamma$. The joint distribution over observation $\mb{y}$ and latent variables $\mb{x}$ and $\bs{\phi}$ can be factorized as
	\begin{align} \label{eq:con:gau:mf-2}
		p(\mb{y},\mb{x},\bs{\phi},\gamma) = p(\mb{y}\!\mid\!\mb{x},\bs{\phi},\gamma) p(\mb{x})p(\bs{\phi})p(\gamma),
	\end{align}
	where $p(\mb{y}\,|\,\mb{x},\bs{\phi},\gamma) = \mc{CN}\big(\bs{\Phi}^H\mb{y};\mb{H}\mb{x} ,\gamma^{-1}\mb{I}_M\big) $. We approximate the posterior distribution using a mean-field
factorization
	\begin{align}
		p(\mb{x},\bs{\phi},\gamma|\mb{y};\mb{H}) \approx  q(\mb{x},\bs{\phi},\gamma) = \prod_{i=1}^K q(x_i)  \prod_{m=1}^M q({\phi}_m)q(\gamma).
	\end{align}
	
		\emph{1) Update $\phi_m$:} Taking the expectation of the conditional in \eqref{eq:con:gau:mf-2} w.r.t. all latent variables except $\phi_m$, the variational distribution $q(\phi_m)$ can be derived as
	\begin{align*}
		q(\phi_m) &\propto 	p(\phi_m)\exp\bigl(\Re\{\eta_{\mr{r},m}^*\e^{\ji \phi_m}\}\bigr),
	\end{align*}
	where $\eta_{\mr{r},m} $ is now defined as
	\begin{align}
		\eta_{\mr{r},m} &= 2\hat{\gamma}y_m\big(\mb{H}_{m,:}\hat{\mb{x}}\big)^*  \nonumber \\
		&= 2\hat{\gamma} \big(|y_m|^2 d_{\mr{r},m} \e^{\ji\hat{\phi}_m} - y_mr_m^*\big). 
	\end{align} 
	
\emph{2) Update $x_i$:} The variational distribution $q(x_i)$ can be derived as
		\begin{align} \label{eq:q-x}
			q(x_i) 
			&\propto p(x_i) \exp\bigl(-\hat{\gamma}\|\mb{h}_i\|^2|x_i-z_i|^2\bigr) \nonumber \\
			&\propto p(x_i) \,\mc{CN}\big(z_i;x_i,1/(\hat\gamma\|\mb{h}_i\|^2)\big).
		\end{align} 
		Thus, $q(x_i) = p\big(x_i\!\mid\! z_i = x_i + \mc{CN}\big(0,1/(\hat\gamma\|\mb{h}_i\|^2)\big)$. The variational  mean $\hat{x}_i$ and variance $\tau_{x_i}$ of $x_i$ can be obtained from Corollary \ref{cor:meanvar}. Here, the matched-filter output is
	\begin{align} 
		z_i &= \frac{\mb{h}_i^H}{\|\mb{h}_i\|^2}\left(\hat{\bs{\Phi}}^H\mb{y} - \sum_{j\neq i}^K \mb{h}_j\hat{x}_j\right)\nonumber \\
		&= \hat{x}_i + {\mb{h}_i^H\mb{r}}/{\|\mb{h}_i\|^2} .
	\end{align} 
	
	\emph{3) Update $\gamma$:} The variational distribution $q(\gamma)$ can be obtained as
\begin{align}
	q(\gamma) &\propto p(\gamma)\exp\left\{\ln p(\mb{y}|\mb{x},\bs{\phi},\gamma)\right\} \nonumber\\
	&\propto \exp\big\{M\ln\gamma - \gamma\blr{\|\bs{\Phi}^H\mb{y} - \mb{H}\mb{x}\|^2 \nonumber\\
    &\quad+ (a_\gamma - 1)\ln \gamma - b_\gamma  \gamma}\big\}.
\end{align}
Thus, $q(\gamma)$ is Gamma with mean
\begin{align} \label{eq:gamma:VM-2}
	\hat{\gamma} = \frac{ M + a_\gamma}{b_\gamma + \blr{\|\bs{\Phi}^H\mb{y} - \mb{H}\mb{x}\|^2} }.
\end{align} 
where
\begin{align}
	&\blr{\|\bs{\Phi}^H\mb{y} - \mb{H}\mb{x}\|^2}   \nonumber\\
	&= \|\mb{r}\|^2 + \sum_{m=1}^M \var[\e^{\ji\phi_m}y_m] + \sum_{i=1}^K\|\mb{h}_i\|^2\var[x_i]\nonumber\\
	&= \|\mb{r}\|^2 + \sum_{m=1}^M|y_m|^2(1-d_{\mr{r},m}^2)  + \sum_{i=1}^K \|\mb{h}_i\|^2 \tau_{x_i}.
\end{align}

The improved MF-VB algorithm using the prior of $x_i$ is summarized in Algorithm~\ref{algo-4}. A key structural advantage over Algorithms \ref{algo-2} and \ref{algo-3} is that the update of $x_i$ in Step 9 does not require phase-noise tracking variables $d_{\rm{t},i},\hat{\theta}_i$; instead, these are encoded within the composite posterior of $x_i$. Consequently, the residual vector update in Step 11 simplifies to $\mb{r} \gets \mb{r} + \mb{h}_i (\hat{x}_i^\ell - \hat{x}_i^{\ell+1})$, reducing per-iteration overhead compared with Algorithm \ref{algo-2}. Final symbol and phase decisions are recovered from the posteriors in parts \eqref{eq:ps} and \eqref{eq:ptheta} of Theorem 1, respectively.
\begin{algorithm}[!t]
	\small
	\caption{-- \textit{\textbf{Improved MF-VB Algorithm} for Inferring $\mb{x}$}}
	\label{algo-4}
	\begin{algorithmic}[1]
		\State \textbf{Input:} $\mb{y}$, $\mb{H}$, and priors $\big\{p(s_i)\big\}$, $\big\{p(\theta_i)\big\}$, $\big\{p(\phi_m)\big\}$\;
		\State \textbf{Output:} $\hat{\mb{s}}$, $\hat{\bs{\theta}}$, and $\hat{\bs{\phi}}$\;
		\State Initialize $\hat{x}_i^1 = 0$, $\tau_{x_i} = \mr{Var}_{p(s_i)}[s_i]$, $\forall i$; and $\mb{r} = \mb{y}$
		\State Initialize $\hat{\phi}_m^1 = 0$, 
		and $d_{\mr{r},m}^1 = {I_1( |\varepsilon_{\mr{r},m}|)}/{I_0(|\varepsilon_{\mr{r},m}|)}$,~$\forall m$
		\For{$\ell=1,2,\ldots,T$}
		\State $\hat{\gamma} \leftarrow M/\bigg(\|\mb{r}\|^2 + \sum\limits_{m=1}^N|y_m|^2\big(1-(d_{\mr{r},m}^\ell)^2\big) + \sum\limits_{i=1}^K \|\mb{h}_i\|^2 \tau_{x_i}^\ell \bigg)$
		\For{$i=1,2,\ldots,K$}
		\State $\displaystyle z_i \gets \hat{x}_i^\ell +   {\mb{h}_i^H\mb{r}}/{\|\mb{h}_i\|^2}$
		\State $\phantom{z_i}\mathllap{\hat{x}_i^{\ell+1}} \gets \mathbb{E}\big[x_i\!\mid\!z_i=x_i+\mc{CN}\big(0,1/(\hat{\gamma}\|\mb{h}_i\|^2)\big)\big]$
		\State $\phantom{z_i}\mathllap{\tau_{x_i}^{\ell+1}} \gets \mr{Var}\big[x_i\!\mid\!z_i=x_i+\mc{CN}\big(0,1/(\hat{\gamma}\|\mb{h}_i\|^2)\big)\big]$
		\State $\phantom{z_i}\mathllap{\mb{r}} \gets \mb{r} + \mb{h}_i (\hat{x}_i^\ell - \hat{x}_i^{\ell+1})$
	\EndFor
	\State Similar to Steps (19) -- (23) of Algorithm \ref{algo-2}.
	\EndFor
	\State $\forall i: \hat{s}_i \leftarrow \arg \max_{a\in \mathcal{A}} p(a\!\mid\!z_i)$ using \eqref{eq:ps} of Theorem \ref{thm:main}\;
	\State $\forall i: \hat{\theta}_i \leftarrow \mathbb{E}[\theta _i\!\mid\! z_i]$ using \eqref{eq:ptheta}  of Theorem \ref{thm:main}\; 
\end{algorithmic}
\end{algorithm}

\subsection{Computational Complexity Analysis}
We end this section with a complexity analysis of the MF-VB, LMMSE-VB, and improved MF-VB algorithms, as shown in Table \ref{tab:complexity}. 

\begin{table}[!t]
\centering
\caption{Algorithm Computational Complexity}
\label{tab:complexity}
\begin{tabular}{|c|c|} 
 \hline
 \bf{Algorithm} & \bf{Complexity} \\ 
 \hline\hline
 Naive ML & $\mc{O}(M|\mc{S}|^K)$  \\ \hline
Naive LMMSE-VB \cite{nguyen2022variational} & $\mc{O}(M^3T + KT|\mc{S}|)$  \\ \hline
SIW in \cite{combes2017approximate} & $\mc{O}(M^3 + M|\mc{S}|^K)$  \\ \hline
MF-VB (Algorithm \ref{algo-2}) & $\mc{O}(KT(M+|S|))$  \\ \hline
LMMSE-VB (Algorithm \ref{algo-3}) & $\mc{O}(M^3 T + KT(M^2+|\mc{S}|))$ \\ \hline
Improved MF-VB (Algorithm \ref{algo-4}) & $\mc{O}(KT(M+|\mc{S}|))$  \\ \hline
\end{tabular}
\end{table}

The naive ML detector performs an exhaustive search over all $|\mc{S}|^K$ candidate vectors, computing the residual $\|\mb{y} - \mb{Hx}\|^2$ for each candidate, yielding complexity $\mc{O}(M|\mc{S}|^K)$.
The naive LMMSE-VB detector requires computing the $K\times K$ Gram matrix $\mb{H}^H\mb{H}$ and its inverse, yielding $\mc{O}(M^3T + KT|\mc{S}|)$, where $T$ is the maximum number of iterations. 
The SIW detector \cite{combes2017approximate} avoids explicit PN tracking but constructs an approximate likelihood by whitening the PN-induced self-interference covariance, followed by two rounds of nearest-neighbor detection via exhaustive search. This incurs a complexity of $\mc{O}(M^3+M|\mc{S}|^K)$, where the $M^3$ term comes from the Cholesky decomposition of the whitening matrix and the $M|\mc{S}|^K$ term reflects the exhaustive search over all $|\mc{S}|^K$ candidate vectors, making SIW computationally intractable for large MIMO systems with high-order modulations.
The MF-VB algorithm (Algorithm \ref{algo-2}) performs, for each user, matched-filter projections $\mb{h}_i^H\mb{r}$ in steps 9 and 13 with cost $\mc{O}(M)$, residual updates in steps 12 and 17 with cost $\mc{O}(M)$, scalar von Mises updates in steps 14, 15, and 16 with cost $\mc{O}(1)$, and posterior table lookups over $|\mc{S}|$ constellation points in steps 10 and 11. Summing over $K$ users and $T$ iterations, the dominant complexity is $\mc{O}(KT(M+|S|))$. 
In the proposed LMMSE-VB algorithm (Algorithm \ref{algo-3}), the invert of the resulting $M\times M$ effective noise covariance matrix $\hat{\bs{\Gamma}}$ with complexity $\mc{O}(M^3)$, which dominates the computation of $\mb{H}\bs{\Sigma}_{\bs{\Theta}\mb{s}}^\ell\mb{H}^H$, whose complexity is $\mc{O}(KM^2)$. The LMMSE-weighted projection $\mb{h}_i^H\hat{\mb{\Gamma}}\mb{r}$ in step 11 has complexity $\mc{O}(KM^2)$. The inner loop in steps 10 $\rightarrow$ 20 has complexity $\mc{O}(KM^2 + K|\mc{S}|)$ per iteration. Therefore, the total complexity over $T$ iterations is $\mc{O}(M^3 T + KT(M^2+|\mc{S}|))$.
The improved MF-VB algorithm (Algorithm \ref{algo-4}) absorbs the transmitter PN into $x_i$, eliminating the separate estimation steps for $\theta_i$ in Algorithm 1. For each user, step 8 incurs complexity $\mc{O}(M)$, while steps 9 and 10 evaluate over $|\mc{S}|$ constellation points, resulting in complexity $\mc{O}(|\mc{S}|)$. The residual update simplifies to $\mb{r} \gets \mb{r} + \mb{h}_i(\hat{x}_i^\ell - \hat{x}_i^{\ell+1})$ with cost $\mc{O}(M)$, removing two von Mises update steps per user per iteration relative to Algorithm \ref{algo-2}. The dominant complexity of Algorithm \ref{algo-4} remains $\mc{O}(KT(M+|\mc{S}|))$. However, Algorithm \ref{algo-4} exhibits longer runtime than Algorithm \ref{algo-2} in simulations, since computing the posterior of the composite variable $x_i$ incurs a higher constant factor than computing the posteriors of $s_i$ and $\theta_i$ separately.

\section{Simulation Results} \label{sec:Sim}
We consider a massive MIMO system where the BS is equipped with $M=24$ antennas and $K = 8$ users with 16-QAM signaling. 
We implement all the iterative algorithms with a maximum of $100$ iterations and consider scenarios with $100$ transmitted data symbols. The noise variance $N_0$ is set based on the SNR, which is defined as
\begin{align}
    {\rm SNR} = \frac{\mathbb{E}[\|\bs{\Phi}\mb{Hx}\|^2]}{\mathbb{E}[\|\mb{n}\|^2]} =\frac{\sum_{i=1}^K \tr\{\mb{R}_i\}}{MN_0}= \frac{K}{M N_0}.
\end{align}
We consider two-channel models for the channel matrix $\mathbf{H}$. The first assumes i.i.d. Gaussian entries, i.e., $\mb{R}_i=(1/M)\mb{I}_M$, which corresponds to i.i.d. Rayleigh fading.  The second considers correlated Gaussian coefficients, corresponding to correlated Rayleigh fading, where each column of $\mb{H}$ follows an exponential spatial correlation model. In this case, the covariance matrix $\mb{R}_i$ is given by
\begin{align}
    [\mathbf{R}_i]_{k\ell} =
\begin{cases}
(1/M)\alpha^{k-\ell}, & \text{if } k \ge \ell, \\
(1/M)\left(\alpha^{\ell-k}\right)^{*}, & \text{if } k < \ell ,
\end{cases}
\end{align}
where $\alpha$ denotes the (complex) correlation coefficient between neighboring receive antennas. In all simulations involving correlated channels, we set $\alpha = 0.4+0.4\,\ji$. The PNs are assumed to be Gaussian with zero mean and variance $\gamma_{\mr{t},i}^{-1}$ for user $i$ and $\gamma_{\mr{r},m}^{-1}$ for receive antenna $m$ at the BS. We compare the proposed VB algorithms with the following baseline algorithms:
\begin{itemize}
    \item Naive ML: The ML detector without considering phase noise.
    \item Naive LMMSE-VB: The LMMSE-VB in \cite{nguyen2022variational} without considering phase noise.
    \item The SIW algorithm in \cite{combes2017approximate}.
\end{itemize}

\begin{figure}[!t]
  	\centering
 	\includegraphics[width=1.0\linewidth]{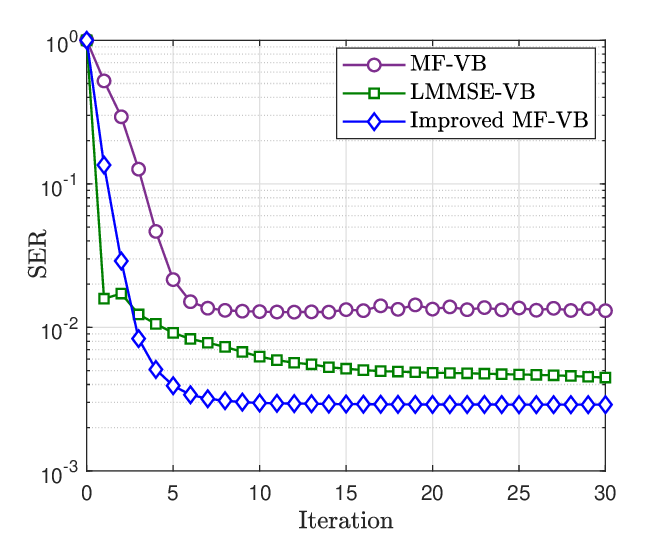}
 	\caption{Convergence of the proposed iterative algorithms assuming i.i.d. Rayleigh fading channels with $M=16$, $K = 8$, $16$-QAM, $\gamma_{\mr{t},i}^{-1/2} = \gamma_{\mr{r},m}^{-1/2} = 5^\circ$. The plots are obtained by averaging over 1000 trials.}
 	\label{Convergence}
\end{figure}
We first study the convergence of all iterative algorithms, such as Algorithms 1, 2, and 3, {with standard deviations of PN $\gamma_{\mr{t},i}^{-1/2} = \gamma_{\mr{r},m}^{-1/2} = 5^\circ$, as shown in Fig.~\ref{Convergence}.} All three VB algorithms are observed to converge reliably within fewer than $20$ iterations, confirming the practical effectiveness of the coordinate ascent variational inference updates employed in the proposed VB framework.

\begin{figure}
  	\centering
 	\includegraphics[width=1.0\linewidth]{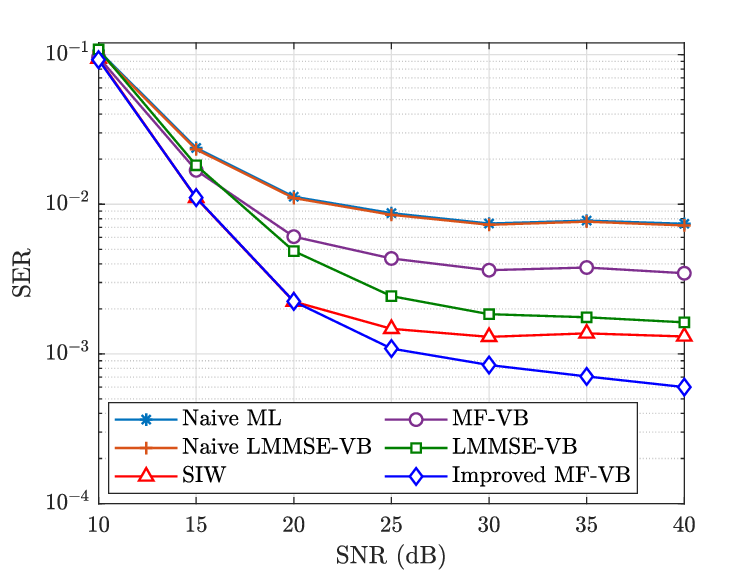}
 	\caption{SER vs. SNR, assuming i.i.d. Rayleigh fading channels with $M=24$, $K = 8$, $16$-QAM signaling, $\gamma_{\mr{t},i}^{-1/2} = \gamma_{\mr{r},m}^{-1/2} = 6^\circ$.}
 	\label{SER_SNR}
\end{figure}

Fig. \ref{SER_SNR} compares the SER performance of all considered algorithms under i.i.d. Rayleigh fading channels. 
The naive ML and naive LMMSE-VB algorithms exhibit the worst performance and quickly reach an error floor around $10^{-2}$ at high SNR, indicating their inability to mitigate PN distortion. In contrast, the proposed VB-based algorithms achieve significant performance improvements by accounting for PN uncertainty. Among them, LMMSE-VB outperforms MF-VB because it employs the full postulated noise covariance matrix $\mb{C}^{\rm post}$ rather than a scalar noise variance, thereby better capturing the residual interference structure. The improved MF-VB algorithm achieves the best SER performance, attaining SERs on the order of $10^{-3}$ at high-SNRs. This is because the reformulation of the transmitter phase noise directly into the composite transmit signal $x_i$, which eliminates the multiplicative coupling between $s_i$ and $\theta_i$ in the VB updates and yields a conjugate posterior under the von Mises prior, as derived in Theorem \ref{thm:main}. The SIW algorithm, achieving similar SER to the improved MF-VB at SNR $< 20$ dB, suffers an increasing performance gap relative to the improved MF-VB at high SNR due to its reliance on a first-order approximation of the PN.

\begin{figure}
  	\centering
 	\includegraphics[width=1.0\linewidth]{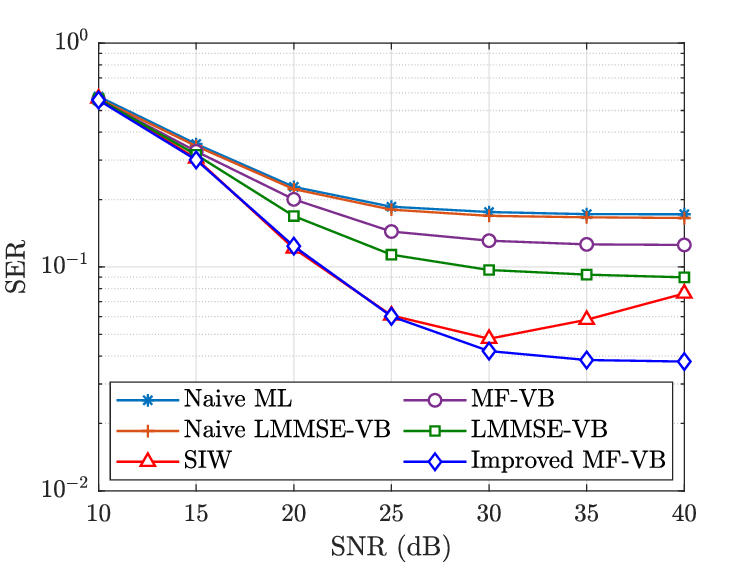}
 	\caption{SER vs. SNR, assuming i.i.d. Rayleigh fading channels with $M=24$, $K = 8$, $64$-QAM signaling, $\gamma_{\mr{t},i}^{-1/2} = \gamma_{\mr{r},m}^{-1/2} = 6^\circ$.}
 	\label{SER_SNR_64QAM_iid}
\end{figure}

In Fig. \ref{SER_SNR_64QAM_iid}, we evaluate the SER performance of all algorithms with a higher-order modulation of 64-QAM under i.i.d. Rayleigh fading. The higher constellation density renders the system more susceptible to phase noise, as the reduced inter-symbol distances make the received signal more sensitive to phase rotations. The naive ML and naive LMMSE-VB detectors again saturate at error floors well above those of the phase-noise-aware methods, confirming that phase noise awareness is indispensable for high-order modulation schemes. The improved MF-VB achieves the lowest SER at high SNRs, followed by SIW, LMMSE-VB, and MF-VB. The performance advantage of the improved MF-VB over the SIW detector becomes even more pronounced under $64$-QAM, highlighting the importance of an accurate posterior model for the transmit signal when higher-order constellations are employed. This result highlights the scalability of the proposed approach to scenarios with high spectral efficiency requirements.

\begin{figure}
  	\centering
 	\includegraphics[width=1.0\linewidth]{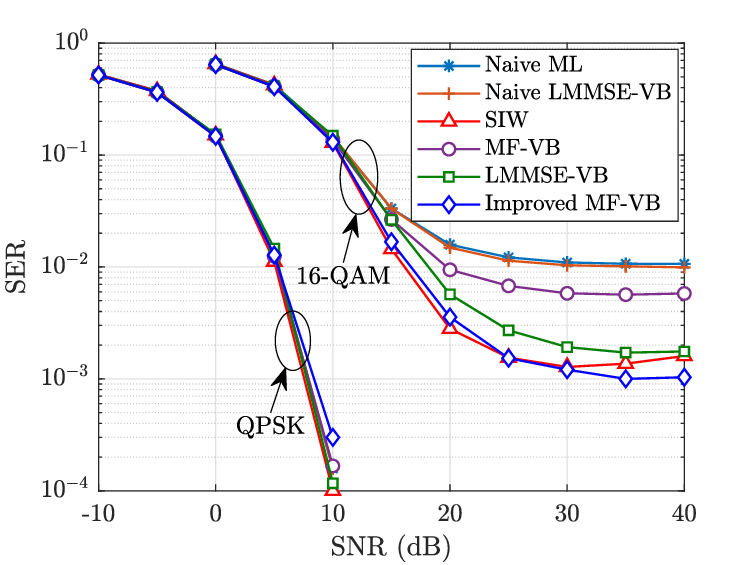}
 	\caption{SER vs. SNR, assuming correlated Rayleigh fading channels with $M=24$, $K = 8$, QPSK and $16$-QAM signaling, $\gamma_{\mr{t},i}^{-1/2} = \gamma_{\mr{r},m}^{-1/2} = 6^\circ$.}
 	\label{SER_SNR_QPSK_16QAM_correlated_channel}
\end{figure}
Fig. \ref{SER_SNR_QPSK_16QAM_correlated_channel} examines the SER performance under correlated Rayleigh fading channels with $\alpha = 0.4+\ji 0.4$. For QPSK, thanks to the high PN tolerance of low-order constellations whose decision regions occupy large phase angles, the SER of all algorithms is nearly identical. For 16-QAM, the improved MF-VB and SIW remain the top-performing detectors, with improved MF-VB achieving superior SER at high SNR. The LMMSE-VB algorithm, which accounts for the full noise covariance structure, exhibits noteworthy gains over MF-VB under correlated fading, since the off-diagonal terms of the residual interference matrix become significant when the channel vectors are correlated. These results confirm that the proposed VB framework is robust to realistic spatial channel correlation models.

\begin{figure}
  	\centering
 	\includegraphics[width=1.0\linewidth]{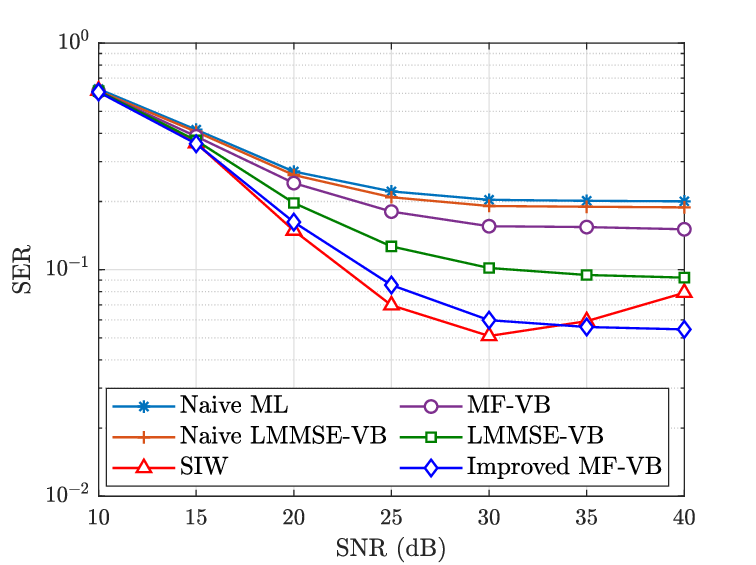}
 	\caption{SER vs. SNR, assuming correlated Rayleigh fading channels with $M=24$, $K = 8$, $64$-QAM signaling, $\gamma_{\mr{t},i}^{-1/2} = \gamma_{\mr{r},m}^{-1/2} = 6^\circ$.}
 	\label{SER_SNR_64QAM_correlated_channel}
\end{figure}
Fig. \ref{SER_SNR_64QAM_correlated_channel} extends the correlated fading evaluation to $64$-QAM under the same system configuration as Fig. \ref{SER_SNR_QPSK_16QAM_correlated_channel}. The combination of spatial correlation and high-order modulation amplifies the performance differences between algorithms accordingly. The naive ML and naive LMMSE-VB detectors exhibit pronounced error floors. The SIW algorithm experiences significant degradation since its approximate likelihood formulation is less effective under correlated channels. In contrast, the improved MF-VB detector maintains a substantial performance advantage, owing to its principled VB treatment of the composite transmit signal and its joint estimation of the noise precision parameter $\gamma$, which adapts to the effective interference level. These results highlight the robustness of the proposed approach to both high-order modulation and spatially correlated propagation environments.

Since the PN is pronounced at mmWave frequencies due to the increased oscillator instability at high carrier frequencies, we also evaluate the proposed VB algorithms under a sparse mmWave propagation environment to assess their practical relevance in this regime. We consider a few-path mmWave propagation environment based on the Saleh-Valenzuela model \cite{nguyen2025mimo}. The channel model is similar to Section II-A of \cite{nguyen2025mimo}, where we set the width of the angular sector as $80^\circ$, with the angle of arrival drawn uniformly from the interval $[-40^\circ,40^\circ]$, the number of paths as $16$, and the antenna spacing as $1/2$.

\begin{figure}
  	\centering
 	\includegraphics[width=1.0\linewidth]{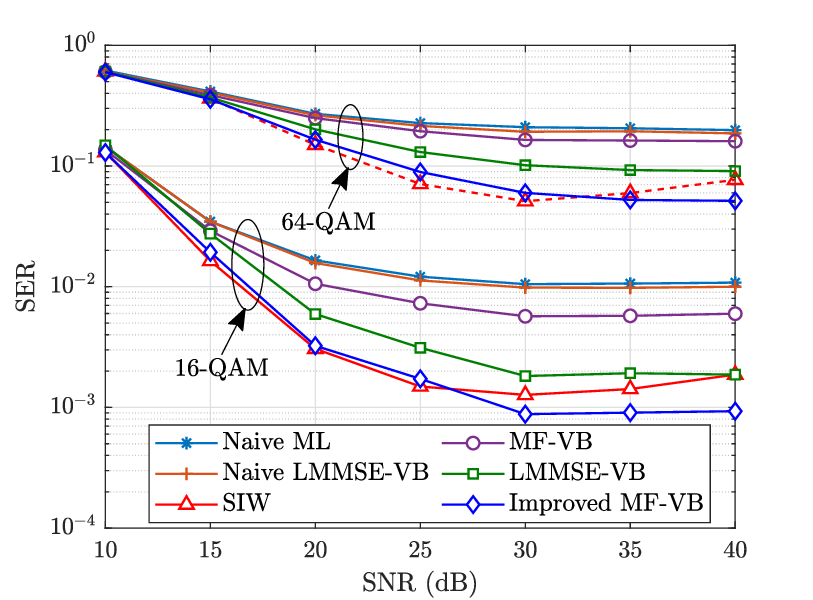}
 	\caption{SER vs. SNR, assuming few-path mmWave channels with $M=24$, $K = 8$, $16$-QAM and $64$-QAM signalings, and $\gamma_{\mr{t},i}^{-1/2} = \gamma_{\mr{r},m}^{-1/2} = 6^\circ$.}
 	\label{SER_SNR_few_path_16_64_QAM}
\end{figure}
Fig. \ref{SER_SNR_few_path_16_64_QAM} evaluates the proposed algorithms under sparse mmWave propagation for 16-QAM and 64-QAM. The improved MF-VB maintains the lowest SER across the entire SNR range in both cases, confirming that its principled VB treatment of the composite transmit signal $x_i$, and adaptive noise precision estimation remain effective under the spatial characteristics of mmWave channels. Compared to the i.i.d. Rayleigh case, LMMSE-VB narrows its gap with the improved MF-VB, since the low-rank interference structure of the sparse mmWave channels allows the full covariance estimate $\hat{\bs{\Gamma}}$ to better capture the dominant interference directions. Conversely, SIW suffers more degradation than in the i.i.d. case, since its approximate likelihood derived under a flat-frequency interference assumption is poorly matched to the spatially concentrated interference of sparse channel geometries. This effect is further amplified under 64-QAM. These results confirm that the advantages of the proposed VB framework are not confined to rich-scattering environments but extend robustly to the sparse, directional channels expected in practical mmWave deployments.


\begin{figure}
  	\centering
 	\includegraphics[width=1.0\linewidth]{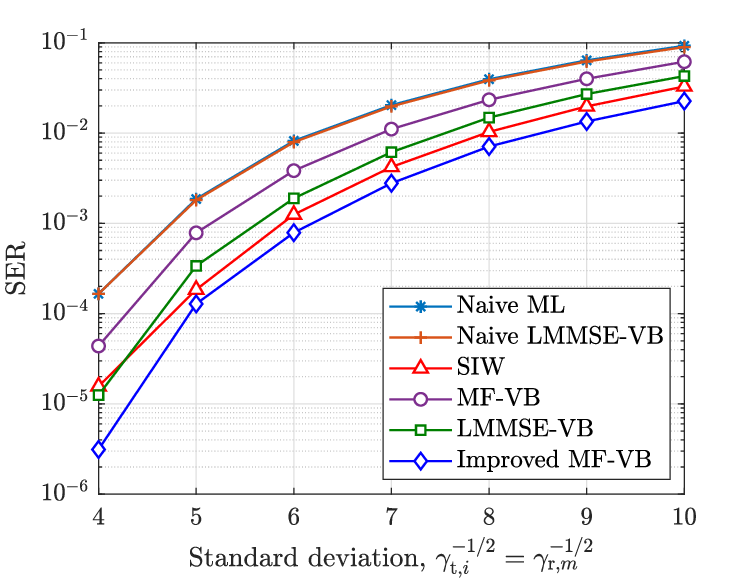}
 	\caption{SER vs. Standard deviation of PN assuming i.i.d. Rayleigh fading channels with $\gamma_{\mr{t},i}^{-1/2} = \gamma_{\mr{r},m}^{-1/2}$, $M=24$, $K = 8$, SNR $=30$ dB, and 16-QAM signaling.}
 	\label{SER_Std}
\end{figure}

Fig. \ref{SER_Std} illustrates the impact of phase noise severity on the detection performance of all considered algorithms, with the PN standard deviation $\gamma_{\mr{t},i}^{-1/2} = \gamma_{\mr{r},m}^{-1/2}$ varied jointly from $4^\circ$ to $10^\circ$ under i.i.d. Rayleigh fading channels. As expected, increasing PN variance degrades the detection accuracy of all methods, since larger phase excursions introduce more severe rotational distortions of the received constellation. The naive ML and naive LMMSE-VB detectors are the most sensitive to increasing PN standard deviation. The improved MF-VB maintains the lowest SER across the entire range of PN standard deviations, demonstrating strong robustness to oscillator quality variations. The reason is that the conjugate structure of the von Mises posterior over $x_i$ continuously adapts the inferred posterior concentration $\kappa_{\mr{t},i}$ in proportion to the observed evidence, and thereby absorbs varying degrees of phase uncertainty without degrading the tightness of the posterior distribution. The SIW algorithm performs close to the improved MF-VB at low PN variance but experiences a widening gap as PN severity increases, consistent with the increase in inaccuracy of its first-order phase noise approximation at larger variance values. These results confirm that the proposed improved MF-VB is well-suited to systems operating at high carrier frequencies, such as mmWave and terahertz bands, where oscillator phase noise tends to be more pronounced.

\begin{figure}
  	\centering
 	\includegraphics[width=1.0\linewidth]{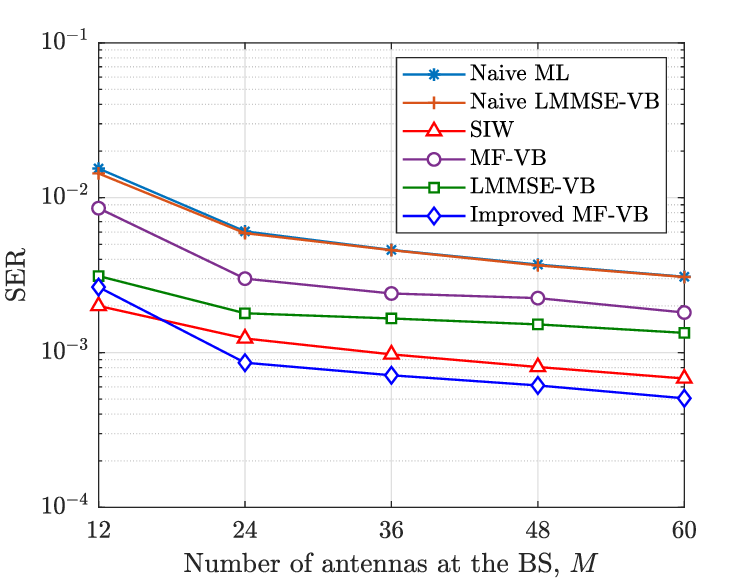}
 	\caption{SER vs. $M$ assuming i.i.d. Rayleigh fading channels with $\gamma_{\mr{t},i}^{-1/2} = \gamma_{\mr{r},m}^{-1/2} = 6^\circ$, $K = 6$, SNR $=30$ dB, 16-QAM signaling.}
 	\label{SER_antennas}
\end{figure}

Fig. \ref{SER_antennas} shows the SER as a function of the number of receive antennas $M$. When $M$ increases from $12$ to $60$, all algorithms benefit from the enhanced spatial diversity and MIMO degrees-of-freedom (DoF). The naive ML and naive LMMSE-VB detectors improve only marginally with increasing $M$, while the improved MF-VB achieves the lowest SER across all antenna settings and maintains its performance advantage over SIW and LMMSE-VB as $M$ grows. This behavior highlights its ability to exploit the additional spatial DoF for enhanced phase-noise mitigation. The observed scalability is promising for massive MIMO deployments.

\begin{figure}
  	\centering
 	\includegraphics[width=1.0\linewidth]{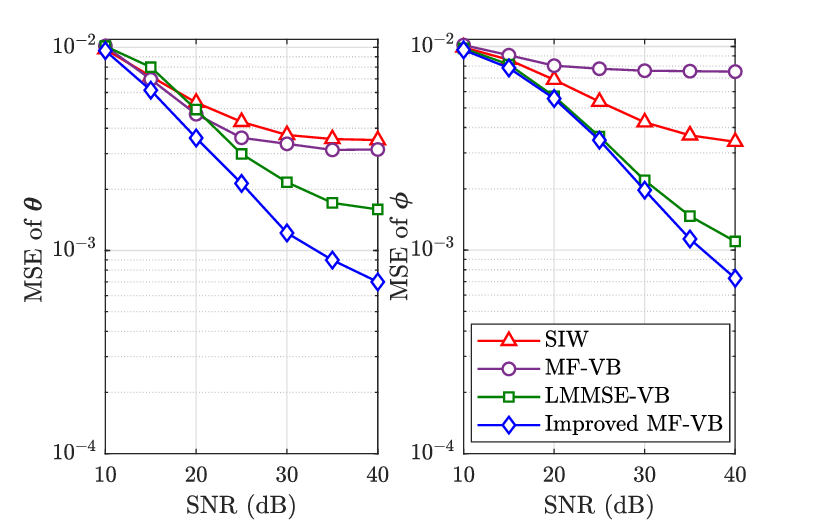}
 	\caption{MSE performance of the phase noise estimation assuming i.i.d. Rayleigh fading channels with $M=24$, $K = 8$, $16$-QAM signaling, and $\gamma_{\mr{t},i}^{-1/2} = \gamma_{\mr{r},m}^{-1/2} = 6^\circ$.}
 	\label{PN_MSEvsSNR}
\end{figure}
Fig. \ref{PN_MSEvsSNR} complements the SER analysis by reporting the MSE of the transmitter PN $\bs{\theta}$ and receiver PN $\bs{\phi}$ as a function of SNR. The improved MF-VB algorithm achieves the lowest MSE for both components across the entire SNR range, with estimation error continuing to decrease steadily at high SNR, as a result of Theorem \ref{thm:main}, where the posterior concentration sharpens proportionally with the noise precision $\hat{\gamma}$. The MF-VB detector saturates at a low MSE level for both PN components, reflecting the estimation penalty of the multiplicative coupling between $s_i$ and $\theta_i$ in its variational updates, while LMMSE-VB partially mitigates this through its full covariance estimate $\hat{\bs{\Gamma}}$. The SIW \cite{combes2017approximate} achieves competitive data detection performance, but it is not designed for PN estimation. As stated in \cite{combes2017approximate}, SIW adopted a modular receiver architecture in which the PN was treated as an unknown parameter and marginalized out to construct the approximate marginal likelihood. The PN was entered into the algorithm only through the self-interference covariance matrix, and no PN estimate was produced during detection. Furthermore, the first-order linearization of PN underlying the SIW approximation introduces an inherent bias in the PN model. In contrast, the proposed VB algorithms maintain and iteratively refine a variational posterior over PN jointly with the data symbols, making them the appropriate framework for joint detection and PN estimation.

\section{Conclusions} \label{sec:Conclusions}
In this paper, we have developed a VB framework for joint phase noise estimation and data detection in uplink multiuser MIMO systems. By treating the PN-corrupted transmit signal as a single composite variable, we bypassed the multiplicative coupling inherent in separation-based approaches and obtained exact conjugate closed-form posteriors under the von Mises prior. Building on this, we developed three detectors: MF-VB, LMMSE-VB, and the improved MF-VB, with increasing detection performance. Simulation results confirmed that the improved MF-VB achieved the lowest SER across i.i.d. Rayleigh fading, correlated Rayleigh fading, and mmWave channels, with the performance advantage becoming more pronounced at high SNR, high-order modulations, and stronger phase noise. All proposed algorithms converged reliably within $20$ iterations and scaled well with the number of antennas. Future work will explore extensions to joint channel estimation, Wiener phase noise temporal modeling, and hardware-efficient implementations for massive MIMO receivers.
\balance
\def\baselinestretch{.97}
\bibliographystyle{IEEEtran}
\bibliography{Refs_PhaseNoise}
\end{document}